\documentclass[journal]{IEEEtran}

\usepackage{algorithmic}
\usepackage{algorithm}
\usepackage{array}

\usepackage{subfigure}
\usepackage{amsmath,amsfonts}

\usepackage{textcomp}
\usepackage{stfloats}
\usepackage{url}
\usepackage{verbatim}
\usepackage{graphicx}
\usepackage{cite}
\usepackage{bm}
\usepackage{soul}
\usepackage{arydshln}
\usepackage{booktabs}
\usepackage{pifont}
\usepackage{amssymb}
\usepackage{tabularx}
\usepackage{microtype}
\usepackage{stfloats}
\usepackage{epstopdf}
\usepackage{makecell}
\usepackage{diagbox}

\usepackage[implicit=false]{hyperref}
\usepackage{tikz,xcolor}

\newtheorem{proof}{Proof}

\newtheorem{remark}{Remark}

\newtheorem{proposition}{Proposition}
\newtheorem{define}{Definition}

\begin{document}

\title{SAR/ISAR Imaging in 6G Network}
\author{Yanmo Hu,~\IEEEmembership{Member,~IEEE}, Shuowen Zhang,~\IEEEmembership{Senior~Member,~IEEE},\\ Ross Murch,~\IEEEmembership{Fellow,~IEEE}, Liang Liu,~\IEEEmembership{Fellow,~IEEE}
\vspace{-0.25em}
\thanks{
This paper was presented in part at the 2026 IEEE International Conference on Communications (ICC) \cite{yanmo_conference}.

Y. Hu, S. Zhang, and L. Liu
are with the
Department of Electrical and Electronic Engineering, The Hong Kong Polytechnic
University, Hong Kong, SAR
(e-mail: \{yanmo.hu, shuowen.zhang, liang-eie.liu\}@polyu.edu.hk).

R. Murch is with the Department of Electronic and Computer Engineering, The Hong Kong University of Science and Technology, Hong Kong, SAR
(e-mail: {eermurch}@ust.hk).

}}
\maketitle

\begin{abstract}

Imaging is a crucial sensing function that finds wide applications
in environmental reconstruction, autonomous driving, etc.
However,
the signal processing methods for
existing radio imaging techniques,
such as millimeter wave (mmWave) imaging,
require high-resolution range estimation
enabled by Gigahertz-level or even Terahertz-level bandwidth,
and cannot be applied
in 6G integrated sensing and communication (ISAC) network
with Megahertz-level bandwidth.
This paper proposes
\mbox{two} novel high-resolution
radio
imaging schemes
that can work on the
6G signals
with limited bandwidth -
bandwidth-independent
synthetic aperture radar (BI-SAR),
where the movable base station (BS)
revolves along the static targets by 360 degrees;
as well as
bandwidth-independent inverse synthetic aperture radar (BI-ISAR),
where the BS is static
and the targets revolve along an axis by 360 degrees.
Different from conventional SAR and ISAR counterparts that rely on range estimation,
our proposed imaging schemes
solely utilize Doppler information
to perform imaging
without any range information.
The main technical challenge of our schemes lies in the anisotropic scattering functions over different directions,
which
hinder the coherent synthesis of the backscattered signals from all directions.
We design
an
iterative adaptive approach-based Doppler association (IAA-DA) algorithm
to tackle the above issue.
Moreover,
we also derive the imaging resolution to characterize the reconstruction quality.
Real-world experiments are provided to show the feasibility and the effectiveness of our proposed 6G imaging schemes.

\end{abstract}


\section{Introduction}\label{Section_Introduction}

\subsection{Motivations and Backgrounds}

\IEEEPARstart{I}{ntegrated}
sensing and communication (ISAC)
has been identified as a main usage scenario for
sixth-generation (6G) applications
and has attracted much attention
\cite{10663814,10878492, 11079818}.
In the literature of ISAC,
valuable works have been performed
to study how to leverage
the communication signals for performing localization
\cite{9724258, 11288092, 10948152, 10637442}.
On the other hand,
imaging,
which needs to estimate the accurate positions of all the points on the target,
is a more challenging task compared to localization.
Because radio imaging has potential in plenty of applications such as environment reconstruction,
it is crucial to investigate the potential of achieving sensing function in future 6G networks.

\begin{figure*}[!t]
\centering
\subfigure[]
{\includegraphics[width=3in]{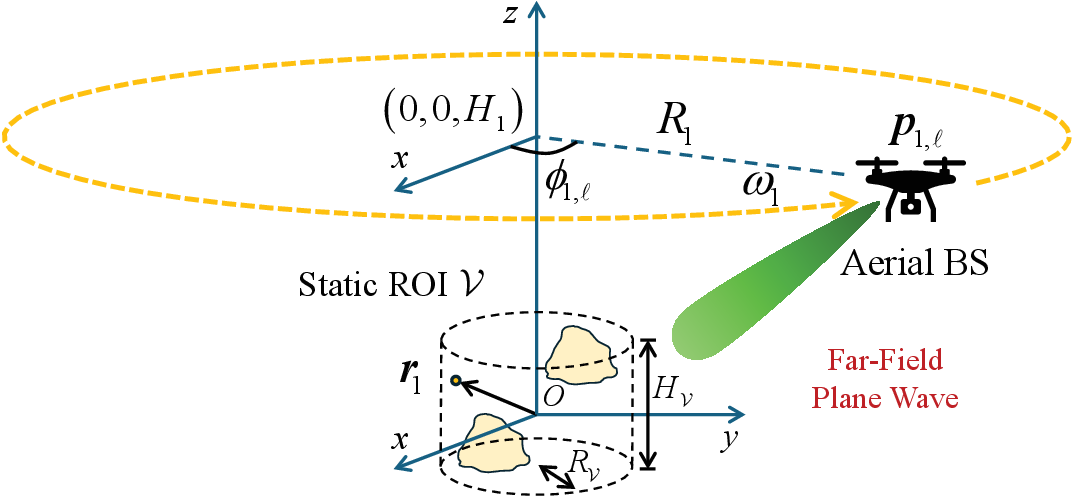}} \ \ \ \ \ \ \ \ \ \ \ \ \
\subfigure[]
{\includegraphics[width=2.9in]{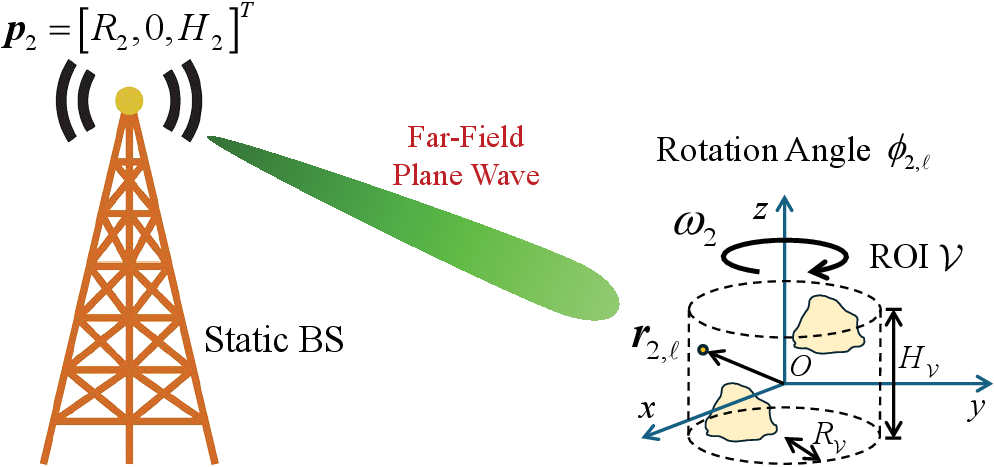}}
\caption{System model of our considered imaging architecture.
(a) BI-SAR imaging mode. (b) BI-ISAR imaging mode.}
\label{Fig_scenario}
\end{figure*}

For scenarios with static anchors and targets,
a fundamental imaging approach is to estimate the range and angle of each grid point on the target \cite{6153063}.
Under this approach, achieving high-resolution imaging requires extremely large bandwidth for accurate range estimation. For example, the bandwidth of the millimeter wave (mmWave) real-time imaging for concealed weapon detection in \cite{942570} is
6 GHz.
However,
such bandwidth is not available in current cellular network.
For example,
the maximal bandwidth in current 5G systems
is 400 MHz
at the mmWave and 100 MHz at the sub-6 GHz band based on \cite{10878492}.
Such bandwidth leads to range resolution in the order of meters and cannot support
high-resolution imaging.
In the literature,
several multi-view schemes have been proposed to achieve bandwidth-independent (BI) imaging without range estimation.
In \cite{gao2025covariancebasedimagingmultiviewfusion},
each BS utilizes multiple antennas to estimate the angle information of each point of the target,
and each point can then be localized based on its angles to several BSs.
Another interesting multi-view imaging scheme
is radio tomographic imaging (RTI) \cite{5374407, 10367810, 10618967, 10947014}.
Leveraging the basic idea from
computed tomography (CT),
RTI
deploys multiple imaging nodes around the region of interest (ROI) \cite{5374407}.
By utilizing attenuation data,
RTI can achieve imaging tasks using minimal resources,
i.e.,
single-frequency, single-antenna, and single-pulse configuration in each imaging node.
However,
the above BI multi-view imaging schemes require a large number of cooperative anchors to fuse their sensing data.
In cellular network,
BSs are typically hundreds of meters away from each other,
and it is hard to receive the echo signals of the targets from many BSs for multi-view imaging.
This gives rise to the problem on how to perform BI imaging using one BS in 6G networks.

With one BS, we can utilize the synthetic aperture radar (SAR) \cite{6504845}
and inverse synthetic aperture radar (ISAR) \cite{4502060} techniques for imaging.
SAR and ISAR have been widely used in radar imaging.
They generate 2D images by using bandwidth
for range resolution
and leveraging Doppler to synthesize a virtual array for cross-range resolution.
Owing to the differences between SAR and ISAR imaging modes,
the strategies for cross-range aperture synthesis are different.
Under SAR imaging,
the targets are static,
and the anchor moves along a trajectory, such that
a virtual array is actively synthesized based on the anchor's movement.
For example,
in conventional stripmap SAR \cite{4373378},
which is the most used linear SAR mode,
the radar moves along a straight path
to synthesize a linear 1D virtual array.
Under ISAR imaging,
the virtual array is passively synthesized by exploiting
the target's own motion.
For example,
in conventional ISAR configurations \cite{532283, 805442},
with small-angle rotation and far-field conditions,
a linear virtual array is formed via Doppler processing.
However, conventional SAR and ISAR schemes cannot directly be
used in
6G cellular networks,
as they rely on wide bandwidth for high-resolution range estimation.
Typical examples include the 3 GHz bandwidth used for through-wall SAR imaging \cite{4497843}, the 1.2 GHz bandwidth for mine-field SAR detection \cite{755021},
and the 1 GHz bandwidth utilized in ISAR systems for noncooperative glider targets \cite{10412121}.
This gives rise to the problem on whether we can modify SAR and ISAR imaging schemes such that they can work without range estimation.

Under the conventional SAR and ISAR imaging schemes,
a virtual linear 1D array is
synthesized via Doppler processing,
thanks to the small-scale movement of the target or the anchor.
According to the array signal processing principle,
a linear 1D array can provide accurate positioning in one dimension \cite{ArraySignalProcessing}.
To generate 2D images,
the conventional SAR and ISAR schemes thus require extremely large bandwidth for range estimation in the second dimension.
To achieve bandwidth-independent imaging, a natural idea is to virtually form a 2D array,
which provides two-dimensional resolution,
for 2D imaging.
Although this requires large-scale, instead of small-scale, movement of the target or the anchor,
this makes it possible to utilize solely the Doppler information, \emph{without any range information}, to perform 2D imaging.
Moreover, for a 2D array,
spatial resolution improves as the angular span increases \cite{317861}.
A full 360-degree view surrounding the target
provides a complete 2D aperture and is able to maximize the resolution.
Such an idea was used in \cite{5374407, 10367810, 10618967, 10947014}.
This motivates us to investigate
the possibility of making 360-degree movement of the target or the BS to achieve BI imaging in 6G cellular network
solely based on the Doppler processing.

\subsection{Contributions}

In this paper, we investigate a fundamental research question:
\emph{Can high-resolution imaging be achieved for targets
when merely the Doppler
information is utilized in 6G systems?}
To answer the above
question,
we propose two novel bandwidth-independent (BI)
imaging techniques in 6G ISAC systems.
The main idea of the BI imaging techniques
is to form a virtual 2D antenna array in
SAR and ISAR modes.
Such schemes perform 2D imaging
\emph{solely based on the Doppler processing}
without requiring any resolution information from bandwidth.
The contributions of this paper are
summarized as follows.
\begin{itemize}

\item
Based on the existing SAR and ISAR imaging schemes, we propose two improved imaging solutions, namely BI-SAR imaging and BI-ISAR imaging. Specifically, as shown in Fig. \ref{Fig_scenario} (a),
under the BI-SAR imaging architecture,
the targets are static, while the BS revolves around
the targets with a fixed height.
In practice, a movable BS, e.g., a UAV or a robot that is equipped
with both the transmit and receive antennas,
can be leveraged to achieve this goal.
Under the BI-ISAR imaging architecture,
the BS is static, while the targets revolve together with a fixed height,
as shown in Fig. \ref{Fig_scenario} (b).
Under our proposed BI-SAR and BI-ISAR schemes, a 2D circular array with 360-degree view is synthesized such that theoretically speaking, merely Doppler processing is sufficient for 2D imaging.

\item
Under our proposed BI-SAR and BI-ISAR schemes,
we propose a complete signal processing algorithm that can fuse the observations from various angles to 2D high-resolution images. The biggest challenge is the anisotropic scattering function,
i.e.,
the scattering energy and phase of any point on the target are different in different directions.
This is quite different from the existing multi-view imaging works
\cite{9534682}, where the scattering function is isotropic such that
any point on the target has a unique scattering value across all directions.
To tackle the challenge caused by anisotropy,
an iterative adaptive approach-based Doppler association (IAA-DA) algorithm
is proposed.
Instead of estimating the scattering value for each grid in each direction,
the proposed algorithm
successfully collects the backscattered signal energy
from all observation angles at each grid
to reconstruct
a 2D map of the total scattering energy within the ROI.

\item
After the signal processing imaging algorithm is proposed,
we also provide theoretical analysis about our BI-SAR and BI-ISAR schemes.
Similar to the relation between bandwidth and range estimation resolution,
we theoretically build the relation between signal frequency and imaging resolution under our schemes.
This provides theoretical guidance to BI imaging in 6G networks.

\item

We conduct real-world experiments to implement our proposed BI-SAR and BI-ISAR schemes on a mmWave ISAC platform.
Experimental results
demonstrate
the superiority and the robustness of the proposed scheme
and
the capability of achieving a \emph{centimeter-level} imaging resolution
using only the Doppler frequency.
To the best of our knowledge, this is the first study to achieve high-resolution imaging using cellular signals in practice.

\end{itemize}

It is worth noting that the idea of creating a 2D virtual array for BI imaging has been studied in the literature before.
Specifically,
under the assumption of isotropic scattering function over all directions,
\cite{4330we23965}
and
\cite{752218}
investigated the BI-ISAR scheme.
However,
their schemes cannot be applied to real systems because targets are usually of complicated shape and with anisotropic scattering function.
On the other hand,
the above idea was utilized in \cite{10475383},
where a user moves along a continuous trajectory
to image the target within the ROI at a reconfigurable intelligent surface (RIS) aided system.
In particular,
\cite{10475383} further considered the practical setup with anisotropic scattering function.
However,
this approach requires multiple measurements
at each user position facilitated by an RIS,
thereby increasing the system complexity.

The rest of this paper is organized as follows.
Section \ref{Section_Signal_Model}
introduces the system models of our proposed BI-SAR and BI-ISAR imaging schemes.
Section \ref{Section_sldfwefbksjdlkfj}
simplifies the signal model via tight approximation.
The IAA-DA imaging algorithm
is proposed in
Section \ref{Section_sdlfjksldkfjlk}.
Section \ref{Section_Algorithm_Performance}
analyzes the imaging resolution of our imaging scheme.
In Section \ref{Section_large_size},
the proposed algorithm is extended to general cases with large-sized ROI.
Experiment results are shown in Section \ref{Section_simulation_experiment},
followed by conclusions in Section \ref{Section_conclusion}.

\section{System Model}\label{Section_Signal_Model}

In this paper,
we consider
the utilization of
communication signals for wireless
imaging in a 6G system,
which consists of one BS and multiple targets,
as shown in Fig. \ref{Fig_scenario}.
All the targets are within a 3D ROI,
which is assumed to be a cylinder
with radius of $R_\mathcal{V}$ meters (m),
height of $H_\mathcal{V}$ m,
and circular ends perpendicular to the $z$-axis.
Moreover, we define the center point of the ROI as the origin point
in our imaging system.
For convenience, denote $\mathcal{V}$ as the above ROI, which consists of all the 3D points in the cylinder.
Additionally,
the BS is equipped with one transmit antenna and one receive antenna.
In the downlink, the BS emits $L$ OFDM symbols to the wireless environment.
Let $N$ and $\Delta f$ Hz denote
the number of sub-carriers and the sub-carrier
spacing of the downlink OFDM signals,
respectively.
Then, the channel bandwidth is $B = N \Delta f$ Hz
and
the length of an OFDM symbol is
$T_0 = {1}/{\Delta f}$ seconds (s).
Furthermore,
a cyclic prefix (CP) with duration of $T_c$ s
is added at the beginning of each OFDM symbol.
Therefore, the duration of each OFDM symbol block
is $\bar{T} = T_0 + T_c$ s,
and the $\ell$-th OFDM symbol block spans the time $\left[\ell {\bar{T}}-{{T}_{c}}, \ell{\bar{T}}+T_0\right]$.
The transmitted signal at time $t$
of the $\ell$-th OFDM symbol block
is given by
\begin{align}
\nonumber {{u}_{\ell }}\left( t \right)={}&{{e}^{j2\pi {{f}_{0}}t}}\sum\limits_{n=0}^{N-1}{{{c}_{n,\ell }}{{e}^{j2\pi n\Delta ft}}}\text{rect}\left( t-\ell {\bar{T}} \right), \\
& \forall t\in \left[\ell {\bar{T}}-{{T}_{c}}, \ell{\bar{T}}+T_0\right], \ \forall \ell = 0,\cdots L-1,
\end{align}
where
$f_0$ denotes the carrier frequency
of the first sub-carrier,
$c_{n, \ell}\in \mathbb{C}$ denotes
the data sample transmitted at the $n$-th sub-carrier
of the
$\ell$-th OFDM symbol,
and
$\text{rect}\left(t\right)=1$
when $t \in \left[-T_c, T_0\right]$
and
$\text{rect}\left(t\right)=0$
otherwise.

The above signals will be reflected by the targets.
Then,
our goal is
to image the targets in the ROI based on their echo signals
received by the BS.
In the current 5G systems,
the maximum bandwidth at FR1 (Sub-6GHz)
and that at FR2 (mmWave)
are 100 MHz and 400 MHz,
leading to range resolution of $1.5\text{m}$ and $0.375\text{m}$, respectively. Therefore, if we adopt the existing imaging methods \cite{4502060, 6504845}, high-resolution imaging is not possible in cellular systems.
In this paper,
we consider two imaging modes that do not rely on range estimation to enable imaging with 6G signals:
BI-SAR imaging mode and BI-ISAR imaging mode.
Specifically, under the BI-SAR imaging architecture,
the targets are static, while the BS revolves around
the targets with a fixed height, as shown in Fig. \ref{Fig_scenario}(a).
In practice, a movable BS, e.g., a UAV or a robot that is equipped
with both the transmit and receive antennas,
can be leveraged to achieve this goal.
Under the BI-ISAR imaging architecture,
the BS is static, while the targets revolve together with a fixed height,
as shown in Fig. \ref{Fig_scenario}(b).
Under both BI-SAR mode and BI-ISAR mode,
we aim to achieve high-resolution imaging of the targets just based on the Doppler information arising from either moving targets or moving BS.
Moreover,
this paper mainly focuses on a challenging imaging setup, where
the distance between
any two points in the ROI, denoted by $\boldsymbol{r}_A$ and $\boldsymbol{r}_B$,
is smaller than the range resolution,
i.e.,
\begin{align}\label{welikjrwlekjr}
{{\left\| {{\boldsymbol{r}}_{A}}-{{\boldsymbol{r}}_{B}} \right\|}_{2}}<\frac{c}{2B}, \ \forall {{\boldsymbol{r}}_{A}},{{\boldsymbol{r}}_{B}} \in \mathcal{V}. \footnotemark
\end{align}
\footnotetext{To make our scheme complete, in Section \ref{Section_large_size}, we will briefly show how to perform imaging when (\ref{welikjrwlekjr}) does not hold.}
$\!\!$In the following, we show
the signal models for both imaging modes,
respectively.

\subsection{Signal Model for BI-SAR Imaging Mode}\label{Section_signal_model_scenario1}

Under the BI-SAR imaging architecture,
as shown in Fig. \ref{Fig_scenario}(a),
we assume that the aerial BS
revolves around the targets by $360{}^\circ$ at a fixed height of $H_1$ m.
More precisely, the trajectory of the aerial BS is a circle
perpendicular to $z$-axis, whose center point and radius
are $\left(0,0,H_1\right)$ and $R_1$ m, respectively.
Moreover,
we denote the angular velocity of the moving aerial BS as $\omega_1$ radian per second,
which is positive when the BS moves in a counterclockwise direction and negative otherwise.
We assume that $R_1$, $H_1$, and $\omega_1$ are all fixed and known.
To revolve by $360{}^\circ$, the overall observation time of the BS for imaging is
$T_1= {2\pi}/{\left|\omega_1 \right|} $ s,
and the number of OFDM symbols used for imaging is $L_1 = \left\lceil {{{T}_{1}}}/{{\bar{T}}}\; \right\rceil $.

In this paper, we use the ``stop-and-go'' approximation
that is widely adopted in radar systems \cite{5420035},
under which the BS is assumed to be at a static position within one OFDM symbol period, while at different positions in different OFDM symbol periods.
Assume that the initial position of the aerial BS is $\boldsymbol{p}_{1,0}=\left[R_1,0,H_1\right]^{T}$.
Then, the position of the BS at the $\ell$-th OFDM symbol
is given by
\begin{align}\label{sddjdjsdkjdj}
{{\boldsymbol{p}}_{1,\ell }}=\boldsymbol{R}\left( {{\phi }_{1,\ell }}  \right)\boldsymbol{p}_{1,0}, \, \forall \ell = 0, \cdots, L_1 - 1,
\end{align}
where
${{\phi }_{1,\ell }}=\omega_1 T_0 \ell$
is the azimuth angle from the aerial BS to the circle center point at the $\ell$-th OFDM symbol,
and
\begin{align}
\boldsymbol{R}\left( \phi  \right)=\left[ \begin{matrix}
   \cos \phi  & -\sin \phi & 0  \\
   \sin \phi  & \cos \phi & 0  \\
   0          &     0     & 1 \\
\end{matrix} \right]
\end{align}
is the rotation matrix corresponding to any azimuth angle $\phi$.

Moreover,
under the standard far-field condition in radar imaging,
i.e., distance between the BS and the ROI
greatly exceeds the size of ROI \cite{4502060},
it is assumed that
the incident wavefront is approximated by a plane wave across the ROI,
i.e.,
both the azimuth and elevation angles of the incident
and backscattered
signals can be considered
equal for all scatterers within the region.
Under the above stop-and-go model and the far-field condition,
the received signal of the BS at the $\ell$-th OFDM symbol
is\footnote{In this paper, we assume that the transmit and receive antennas at the BS is separated by 5 - 10 wavelengths such that the self-interference under the full-duplexing mode can be mitigated \cite{9724258}.
Then, since
the BS knows its transmit signals, the self-interference can
be cancelled in the digital domain efficiently.}
\begin{align}\label{sdlfjslkdfjs}
\nonumber  {\tilde{s}_{1,\ell }}\left( t \right)={}& {{\xi }_{1}}\int_{\mathcal{V}}{{{f}_{3D}}\left( {{\boldsymbol{r}}_{1}};{{\phi }_{1,\ell }} \right){{u}_{\ell }}\left( t-\frac{2{{d}_{1,\ell }}\left( {{\boldsymbol{r}}_{1}} \right)}{c} \right)d{{\boldsymbol{r}}_{1}}} \\
 & +{\tilde{z}_{1,\ell }}\left( t \right), \ \forall t, \ell ,
\end{align}
where
$\xi_{1}\in \mathbb{R}$
denotes the path loss
from the BS to the targets to the BS,
${{{\tilde{z}}}_{1, \ell }}\left( t \right) \in \mathbb{C}$
denotes the additive white Gaussian noise (AWGN) of time $t$ at the $\ell$-th symbol duration,
${{f }_{3D}}\left( \boldsymbol{r}_1;{{\phi }_{1,\ell }} \right) \in \mathbb{C}$
denotes the anisotropic backscattering function at the $\ell$-th OFDM symbol duration, from the position
$\boldsymbol{r}_1 = \left[x, y, z\right]^{T}\in \mathcal{V}$
to the BS with an azimuth angle $\phi_{1, \ell}$\footnote{The scattering function also depends on the elevation angle of the reflection. However, in our consider BI-SAR imaging mode, the elevation angle of the BS is fixed over time and known. Therefore, we omit the elevation angle in the scattering function.},
and
${{d}_{1,\ell }}\left( \boldsymbol{r}_1 \right) =  {{\left\| {{\boldsymbol{p}}_{1,\ell }}-
\boldsymbol{r}_1 \right\|}_{2}}$
represents
the distance between the
BS position at the $\ell$-th OFDM symbol and the position $\boldsymbol{r}_1 \in \mathcal{V}$.
Note that
${{f }_{3D}}\left( \boldsymbol{r}_1;{{\phi }_{1,\ell }} \right)$
is non-zero when $\boldsymbol{r}_1$
is on target area
and is zero otherwise.

Next,
after
the fast Fourier transform (FFT),
the received signal
of the $n$-th sub-carrier at the $\ell$-th OFDM symbol
is given by
\begin{align}\label{sdflkjsklfj}
\nonumber {{s}_{1,n,\ell }}= {}& c_{n,\ell} {{\xi }_{1}}\\
\nonumber  & \! \! \times \int_{\mathcal{V}}{{{f}_{3D}}\left( \boldsymbol{r}_1;{{\phi }_{1,\ell }} \right){{e}^{-j\frac{4\pi }{\lambda }{{d}_{1,\ell }}\left( \boldsymbol{r}_1 \right)}}{{e}^{-j2\pi n\Delta f\frac{2{{d}_{1,\ell }}\left( \boldsymbol{r}_1 \right)}{c}}}d\boldsymbol{r}_1} \\
  & \! \!  +{{{z}}_{1,n,\ell }}, \ \forall n,\ell,
\end{align}
where
${{z}_{1,n,\ell }}$
is the noise at the $n$-th sub-carrier of the $\ell$-th symbol,
and
$\lambda$ is the wavelength of the signal.

\subsection{Signal Model for BI-ISAR Imaging Mode}\label{Section_signal_model_scenario2}

Under the BI-ISAR imaging mode,
we assume that the static BS is located at $\boldsymbol{p}_2 = \left[R_2,0,H_2\right]^T \in \mathbb{R}^{3 \times 1}$,
while the targets in ROI revolve
along the $z$-axis with the same angular velocity,
denoted as
$\omega_2$ radian per second,
which is positive when targets are moving in a clockwise direction and negative otherwise.
We assume that
the position of the BS $\boldsymbol{p}_2$
and
the angular velocity $\omega_2$
are all fixed and known.
Similar to BI-SAR,
all the targets revolve by
$360{}^\circ$,
and the observation time of the BS is
$T_2= {2\pi}/{\left|\omega_2 \right|} $ s,
and the number of OFDM symbols used for imaging is $L_2 = \left\lceil {{{T}_{2}}}/{{\bar{T}}}\; \right\rceil $.
Moreover,
we also adopt the stop-and-go approximation for characterizing the movement of the targets.
Then,
for any point in ROI whose initial position is
$\boldsymbol{r}_{2,0}$,
its position at the $\ell$-th OFDM symbol duration is
\begin{align}
{{\boldsymbol{r}}_{2,\ell }}= \boldsymbol{R}\left({{\phi }_{2, \ell }}  \right)\boldsymbol{r}_{2, 0}, \ \forall \ell = 0, \cdots, L_2 - 1,
\end{align}
where
${{\phi }_{2, \ell }}=-\omega_2 T_0 \ell$
is the
azimuth angle of the incident signals under the far-field condition \cite{4502060}.

Similar to (\ref{sdflkjsklfj}),
the BS's received signal of the $n$-th sub-carrier at the $\ell$-th OFDM symbol
is given by
\begin{align}\label{sdflksjflksjd}
\nonumber    &{{s}_{2,n,\ell }} =c_{n,\ell}{{\xi }_{2}}\\
\nonumber & \times \int_{\mathcal{V}}{{{f}_{3D}}\left( \boldsymbol{r}_{2, 0 };{{\phi }_{2,\ell }} \right){{e}^{-j\frac{4\pi }{\lambda }{{d}_{2,\ell }}\left( \boldsymbol{r}_{2, 0 } \right)}}{{e}^{-j2\pi n\Delta f\frac{2{{d}_{2,\ell }}\left( \boldsymbol{r}_{2, 0 } \right)}{c}}}d\boldsymbol{r}_{2, 0 }} \\
  &  +{{{z}}_{2,n,\ell }}, \ \forall n,\ell,
\end{align}
where
$\xi_{2}\in \mathbb{R}$
denotes the path loss
from the BS to the targets to the BS,
${{f}_{3D}}\left( \boldsymbol{r}_{2,0} ; \phi_{2, \ell}  \right)$
denotes the anisotropic scattering function
at the $\ell$-th OFDM symbol duration, from the point whose
initial position at the 0-th OFDM symbol is
$\boldsymbol{r}_{2,0}$ to the BS with an
azimuth angle
$\phi_{2, \ell}$,
${{d}}_{2,\ell}\left( \boldsymbol{r}_{2,0} \right) = {{\left\| {{\boldsymbol{p}}_{2 }}  -
{{\boldsymbol{r}}_{2,\ell }} \right\|}_{    2}}$
denotes the distance between the point
$\boldsymbol{r}_{2, \ell}$
(whose initial position is $\boldsymbol{r}_{2,0}$)
and the BS at the $\ell$-th symbol, and
${{z}_{2,\ell }}$
denotes the AWGN
at the $\ell$-th symbol.

\subsection{A Unified Signal Model for BI-SAR and -ISAR Imaging}

Intuitively,
the revolvement of the aerial BS
along the $z$-axis
by $\omega$
in the clockwise (counterclockwise) direction under BI-SAR
is equivalent to the revolvement of the targets
along the $z$-axis by $\omega$ in the counterclockwise (clockwise) direction under BI-ISAR.
In the following, we provide a unified signal model of the received signals under BI-SAR imaging mode given in
(\ref{sdflkjsklfj})
and those under BI-ISAR imaging mode given in (\ref{sdflksjflksjd}).
Specifically, under both the BI-SAR and BI-ISAR imaging modes,
the received signals at the $\ell$-th OFDM symbol
depend on the initial position
of the BS
(${{\boldsymbol{p}}_{1,0 }}$ for SAR mode
and $\boldsymbol{p}_2$ for ISAR mode),
the initial position of each scattering point in ROI
($\boldsymbol{r}_1$ for SAR mode
and ${{\boldsymbol{r}}_{2,0 }}$ for ISAR mode),
and the azimuth angle between the BS and the ROI
at each $\ell$-th symbol
($\phi_{1, \ell}$ for SAR mode
and $\phi_{2, \ell}$ for ISAR mode).
By
defining $\boldsymbol{p}=\left[R,0,H\right]^{T}$
as the initial BS position,
$\boldsymbol{r}$ as the initial position of a point in ROI,
and $\phi_{ \ell}$ as the azimuth angle between the BS
and the ROI at the $\ell$-th OFDM symbol,
it can be readily shown that
(\ref{sdflkjsklfj})
under BI-SAR
and
(\ref{sdflksjflksjd})
under
BI-ISAR can be unified as
\begin{align}\label{imaging_model_sdfadf}
\nonumber   {{s}_{n,\ell }} ={}&    c_{n,\ell}{{\xi }}\int_{\mathcal{V}}{{{f}_{3D}}\left( \boldsymbol{r};{{\phi }_{\ell }} \right){{e}^{-j\frac{4\pi }{\lambda }{{d}_{\ell }}\left( \boldsymbol{r} \right)}}{{e}^{-j2\pi n\Delta f\frac{2{{d}_{\ell }}\left( \boldsymbol{r} \right)}{c}}}d\boldsymbol{r}} \\
  &  +{{{z}}_{n,\ell }}, \ \forall n,\ell,
\end{align}
where
$\xi$
denotes the path loss
from the BS to the targets to the BS,
and
${{d}_{\ell }}\left( \boldsymbol{r} \right) =  {{\left\| \boldsymbol{R}\left( {{\phi }_{\ell }}  \right)\boldsymbol{p} -
\boldsymbol{r} \right\|}_{2}}$
represents
the distance between the
BS and the position $\boldsymbol{r} \in \mathcal{V}$
at the $\ell$-th symbol.

In following sections,
we
will show how to perform BI imaging based on the unified signal model (\ref{imaging_model_sdfadf}).

\vspace{0.5em}
\begin{remark}
Due to abundant bandwidth resources,
conventional SAR and ISAR imaging generally
rely on small-angle observations.
In contrast to these traditional schemes,
our proposed BI imaging leverages 360-degree observations
to overcome the constraints imposed by narrow bandwidth.
In other words,
the deficiency in bandwidth is compensated for by full-angular observation.
\end{remark}
\vspace{0.5em}

\section{Signal Model Simplification}\label{Section_sldfwefbksjdlkfj}
Directly applying the imaging algorithm to (\ref{imaging_model_sdfadf})
may be challenging due to its complicated structure.
As the integral in (\ref{imaging_model_sdfadf}) comprises three distinct terms,
this section introduces three essential approximations to simplify each term.

\subsection{Bandwidth-Limited Approximation}

In this subsection, a bandwidth-limited approximation is employed to simplify
the exponential term associated with the frequency domain,
i.e., ${{e}^{-j2\pi n\Delta f\frac{2{{d}_{\ell }}\left( \boldsymbol{r} \right)}{c}}}$
in (\ref{imaging_model_sdfadf}).

The bandwidth constraint shown in (\ref{welikjrwlekjr})
implies that
all scatterers within the ROI are located within the same range bin.
As a result,
the signal bandwidth
does not offer additional resolution
to distinguish between these scatterers.
Using (\ref{welikjrwlekjr}),
we can readily obtain the following relationship
\begin{align}
\nonumber \left| {{d}_{\ell }}\left( {{\boldsymbol{r}}_{A}} \right) - {{d}_{\ell }}\left( {{\boldsymbol{r}}_{B}} \right) \right| &\le {{\left\| \left(\boldsymbol{R}\left( {{\phi }_{\ell }}  \right)\boldsymbol{p}-{{\boldsymbol{r}}_{A}}\right)
-\left(\boldsymbol{R}\left( {{\phi }_{\ell }}  \right)\boldsymbol{p}-{{\boldsymbol{r}}_{B}}\right) \right\|}_{2}} \\
& <\frac{c}{2B}, \ \forall {{\boldsymbol{r}}_{A}},{{\boldsymbol{r}}_{B}} \in \mathcal{V}.
\end{align}
Since the scatterers cannot be resolved in range using the available bandwidth,
they can be regarded as sharing a common distance from
the BS.
We approximate the range of every scatterer using a single value
\begin{align}\label{welidjfsdlkhgjlk}
{{d}_{\ell }}\left( \boldsymbol{r} \right) \approx {{\left\| \boldsymbol{R}\left( {{\phi }_{\ell }}  \right)\boldsymbol{p} \right\|}_{2}}=  {{\left\| \boldsymbol{p} \right\|}_{2}} =  \sqrt{{{R}^{2}}+{{H}^{2}}}, \
\forall \ell, \boldsymbol{r}\in \mathcal{V}.
\end{align}
Based on this approximation,
(\ref{imaging_model_sdfadf})
can be simplified as
\begin{align}\label{welkjlknmvnds}
\nonumber {{s}_{n,\ell }}\approx {}& {{c}_{n,\ell }}\xi {{e}^{-j2\pi n\Delta f\frac{2\sqrt{{{R}^{2}}+{{H}^{2}}}}{c}}}\\
&\int_{\mathcal{V}}{{{f}_{3D}}\left( \boldsymbol{r};{{\phi }_{\ell }} \right){{e}^{-j\frac{4\pi }{\lambda }{{d}_{\ell }}\left( \boldsymbol{r} \right)}}d\boldsymbol{r}}+{{z}_{n,\ell }},\ \forall n,\ell.
\end{align}

\subsection{Far-Field Approximation}\label{Section_sldfksjdlkfj}

Then, a far-field approximation
is exploited to simplify the exponential term related to the time domain,
i.e., ${{e}^{-j\frac{4\pi }{\lambda }{{d}_{\ell }}\left( \boldsymbol{r} \right)}}$
in (\ref{imaging_model_sdfadf}).

Under the far-field assumption,
where ${{\left\| \boldsymbol{R}\left( {{\phi }_{\ell }}  \right)\boldsymbol{p} \right\|}_{2}}=  {{\left\| \boldsymbol{p} \right\|}_{2}} \gg {\left\|\boldsymbol{r}\right\|}_{2}$
holds for all $\ell$ and $\boldsymbol{r}\in \mathcal{V}$,
the following approximation can be derived
\begin{align}
\nonumber{{\left\| \boldsymbol{R}\left( {{\phi }_{\ell }}  \right)\boldsymbol{p} -\boldsymbol{r} \right\|}_{2}} &\approx {{\left\| \boldsymbol{R}\left( {{\phi }_{\ell }}  \right)\boldsymbol{p} \right\|}_{2}}-\frac{   \left(\boldsymbol{R}\left( {{\phi }_{\ell }}  \right)\boldsymbol{p}\right)^{T}}{{{\left\| \boldsymbol{R}\left( {{\phi }_{\ell }}  \right)\boldsymbol{p} \right\|}_{2}}}\boldsymbol{r}\\
 &= \sqrt{{{R}^{2}}+{{H}^{2}}}-\frac{1}{\sqrt{{{R}^{2}}+{{H}^{2}}}}\boldsymbol{p}^{T} \boldsymbol{R}^{T}\left( {{\phi }_{\ell }}  \right)\boldsymbol{r}.
\end{align}
Applying this approximation to
(\ref{welkjlknmvnds})
yields (\ref{ujkdcgljfg}),
shown at the top of the next page.
\begin{figure*}[!t]
\begin{align}\label{ujkdcgljfg}
 \nonumber {{{{s}}}_{n,\ell }} & \approx {{c}_{n,\ell }}\xi {{e}^{-j2\pi n\Delta f\frac{2\sqrt{{{R}^{2}}+{{H}^{2}}}}{c}}} \int_{\mathcal{V}}{{{f}_{3D}}\left( \boldsymbol{r};{{\phi }_{\ell }} \right){{e}^{-j\frac{4\pi }{\lambda }\left( \sqrt{{{R}^{2}}+{{H}^{2}}}-\frac{1}{\sqrt{{{R}^{2}}+{{H}^{2}}}}\boldsymbol{p}^{T} \boldsymbol{R}^{T}\left( {{\phi }_{\ell }}  \right)\boldsymbol{r} \right)}}d\boldsymbol{r}}+{{{{z}}}_{n,\ell }} \\
 &= {{c}_{n,\ell }}\xi {{e}^{-j2\pi n\Delta f\frac{2\sqrt{{{R}^{2}}+{{H}^{2}}}}{c}}} \iint_{  \mathcal{V}}  {{{f}_{2D}}  \left( x,y;{{\phi }_{\ell }} \right)  {{e}^{j\frac{4\pi }{\lambda } \frac{R}{\sqrt{{{R}^{2}} +  {{H}^{2}}}}  \left( x\cos \omega {{T}_{0}}\ell +y\sin \omega {{T}_{0}}\ell  \right)}}dxdy} +{{{{z}}}_{n,\ell }}, \ \forall n,\ell.
\end{align}
\normalsize
\hrulefill
\end{figure*}
In (\ref{ujkdcgljfg}),
the original 3D scattering function
degenerates into
\begin{align}\label{dslkjfslkjflskdjf}
\nonumber &{{f}_{2D}}\left( x,y;{{\phi }_{\ell }} \right)={{e}^{-j\frac{4\pi }{\lambda }\sqrt{{{R}^{2}}+{{H}^{2}}}}}\\
& \ \ \ \ \ \ \ \ \ \ \ \times \int_{\mathcal{V}}{{{f}_{3D}}\left( x,y,z;{{\phi }_{\ell }} \right){{e}^{j\frac{4\pi }{\lambda }\frac{zH}{\sqrt{{{R}^{2}}+{{H}^{2}}}}}}dz}, \ \forall \ell .
\end{align}

\begin{remark}
(\emph{Imaging Result})
As shown in (\ref{dslkjfslkjflskdjf}),
the $z$-axis component is merged into the scattering function.
Thus,
our proposed imaging scheme yields a 2D image.
Moreover,
we can also conclude that
\emph{the imaging result is the projection of the targets within the ROI onto the $xoy$ plane}.
\end{remark}

\subsection{Small-Angle Approximation for Scattering Function}\label{Section_Signal_Model_Approximation}

Due to the anisotropic backscattering properties,
the scattering function in (\ref{ujkdcgljfg}),
i.e.,
${{f}_{2D}}\left( x,y;{{\phi }}_{\ell} \right)$,
has different values when the ROI is observed from different azimuth angles ${{\phi }}_{\ell}$.
If (\ref{ujkdcgljfg}) is directly adopted for imaging,
the exponential term in (\ref{ujkdcgljfg}) will be absorbed into ${{f}_{2D}}\left( x,y;{{\phi }}_{\ell} \right)$
and become part of the anisotropic scattering,
thereby failing to realize the intended imaging function.
To address this problem,
a reasonable approximation or representation \cite{650078}
for ${{f}_{2D}}\left( x,y;{{\phi }}_{\ell} \right)$
is necessary.
In conventional ISAR \cite{4502060, 10851319},
it is commonly accepted that the scattering function remains relatively stable
when the total rotational angle is sufficiently small.
However,
for our BI imaging,
the ROI is observed over a full angular range,
rather than the small angles typically in ISAR.
Thus, we
propose the following \emph{small-angle approximation}
to
extend the approximation in ISAR imaging to our scheme.

Define $\delta_{SA} \ll 2\pi$
as a sufficiently small angle
preserving the correlation properties of the scattering function with respect to the angle.
We assume that given any ${{\Phi }} \in \left[0, 2\pi\right)$,
we have
\begin{align}\label{sldfkjslkdfjslkfj}
{{f}_{2D}}\left( x,y;{{\phi }} \right)    \approx     {{f}_{2D}}\left( x,y;{{\Phi }} \right), \
\forall {\phi}\in \left[{{\Phi }}-\delta_{SA}, {{\Phi }}+\delta_{SA}\right],
\end{align}
where we refer to the range of ${\phi}$ as the \emph{angular correlation range (ACR)}.
The parameter $\delta_{SA}$ should be chosen appropriately based on target characteristics.
Extensive experiments confirm that it should generally not exceed $3{}^\circ \frac{\pi}{180{}^\circ}\text{rad}$.

To apply the small-angle approximation,
we uniformly sample $K$ azimuth angle centers
over the range of $\left[0, 2\pi\right)$.
The $k$-th
azimuth angle center can be expressed as
${{\Phi }_{k}}=\frac{2\pi}{K}k$,
where
$k=0,\cdots,K - 1$.
Using (\ref{sldfkjslkdfjslkfj}),
the ACR of the $k$-th azimuth angle center
is
\begin{align}
{{\Phi }_{k}} - \delta_{SA}\le \phi \le {{\Phi }_{k}} + \delta_{SA}, \, \forall k.
\end{align}
Thus, when the BS observes the ROI within the above angular range,
according to (\ref{sldfkjslkdfjslkfj}),
the scattering function
is assumed to be constant,
with the value of ${{f}_{2D}}\left( x,y;{{\Phi }} \right)$.

Next,
we map the above ACR to the symbol index.
Note that the scanning angle interval
between two adjacent symbols is $\omega T_0$.
We can calculate
the symbol index corresponding to the $k$-th azimuth angle center $\Phi_k$
by
$\left[\Phi_k/{\left(\omega T_0\right)}\right] \in \mathbb{Z}$,
and
the number of the symbols corresponding to the small angle $\delta_{SA}$
by
$M = \left[{\delta_{SA}}/{\left(\omega T_0\right)}\right] \in \mathbb{Z}$,
where $\left[\cdot\right]$
denotes the round operation.
During the ACR of the $k$-th azimuth angle center,
i.e., ${{\Phi }_{k}}=\frac{2\pi}{K}k$,
the range of the symbol index is thus given by
\begin{align}\label{sdlfkjsldkfj}
 {\left[\frac{\Phi_k}{\omega T_0}\right]-   M   \le   \ell   \le     \left[\frac{\Phi_k}{\omega T_0}\right]   +   M}  , \ \forall \ell \in \mathbb{Z}, \ \forall k.
\end{align}

Now,
we recast the signal model in (\ref{ujkdcgljfg})
into a 3D representation.
Specifically,
based on (\ref{sdlfkjsldkfj}),
we replace $\ell$ with $\left[\Phi_k/{\left(\omega T_0\right)}\right] + m$,
i.e,
$\ell =\left[ \frac{2\pi }{\omega {{T}_{0}}K}k \right]+m$,
where
$m = -M, \cdots, -1, 0, 1, \cdots, M$.
This yields the refined signal representation shown in (\ref{sdlsdkdkjci}) at the top of the next page.
\begin{figure*}[!t]
\begin{align}\label{sdlsdkdkjci}
\nonumber {{S}_{n,m,k}}& = {{c}_{n,\ell }}\xi {{e}^{-j2\pi n\Delta f\frac{2\sqrt{{{R}^{2}}+{{H}^{2}}}}{c}}} \iint_{\mathcal{V}}{{{f }_{2D}}\left( x,y;{{\Phi }_{k}} \right){{e}^{j\frac{4\pi }{\lambda }\frac{{{R}}}{\sqrt{R^{2}+H^{2}}}\left[ x\cos \left( {{\Phi }_{k}}+{{\omega }}{{T}_{0}}m \right)+y\sin \left( {{\Phi }_{k}}+{{\omega }}{{T}_{0}}m \right) \right]}}dxdy}+{{Z}_{n,m,k }} \\
 & \approx  {{c}_{n,\ell }}\xi {{e}^{-j2\pi n\Delta f\frac{2\sqrt{{{R}^{2}}+{{H}^{2}}}}{c}}} \iint_{\mathcal{V}}{{{{\tilde{f }}}_{2D}}\left( x,y;{{\Phi }_{k}} \right)
 {{e}^{j\frac{4\pi }{\lambda }\frac{{{R}}}{\sqrt{R^{2}+H^{2}}}\left( -x\sin {{\Phi }_{k}}+y\cos {{\Phi }_{k}} \right){{\omega }}{{T}_{0}}m}}   dxdy}+{{Z}_{n,m,k }}, \, \forall n,m, k.
\end{align}
\normalsize
\hrulefill
\end{figure*}
In (\ref{sdlsdkdkjci}),
the first-order approximation
is performed
based on the small-angle approximation,
and
\begin{align}
{{{\tilde{f }}}_{2D}}\left( x,y;{{\Phi }_{k}} \right)={{f }_{2D}}\left( x,y;{{\Phi }_{k}} \right){{e}^{j\frac{4\pi }{\lambda }\frac{{{R}}}{\sqrt{R^{2}+H^{2}}}\left( x\cos {{\Phi }_{k}}+y\sin {{\Phi }_{k}} \right)}}.
\end{align}

\subsection{Signal Model Analysis}\label{Section_Analysis_Scenario}

In this subsection,
we examine the Doppler frequency and the unambiguous imaging region
of the signal model in (\ref{sdlsdkdkjci}).

\vspace{0.7em}

\subsubsection{Doppler Frequency}
First,
we have
the following proposition to characterize the Doppler frequency for a scatterer.
\vspace{0.5em}
\begin{proposition}\label{proposition_sldkfjslkdj}
(\emph{Doppler Frequency})
The Doppler frequency of the scatterer
is
found to be sinusoidally modulated
with respect to the azimuth angle.
For a given azimuth angle $\Phi_k$,
the Doppler frequency corresponding to a scatterer
at position
$\left(x, y\right)$ is expressed as
\begin{align}\label{weliwjleikjr}
 {{D}_f}\left( x,y;{{\Phi }_{k}} \right)=\frac{2\omega }{\lambda }\frac{1}{\sqrt{1+{{\frac {H^{2}}{R^{2}}}}}}\left( -x\sin {{\Phi }_{k}}+y\cos {{\Phi }_{k}} \right), \ \forall k.
\end{align}
\end{proposition}
\begin{proof}
For the scatterer
at position $\left(x, y\right)$,
the signal component in (\ref{sdlsdkdkjci})
is
${{y}_{m,k}}\left| _{\left( x,y \right)} \right.={{\tilde{f }}_{2D}}\left( x,y;{{\Phi }_{k}} \right){{e}^{j\frac{4\pi }{\lambda }\frac{R}{\sqrt{{{R}^{2}}+{{H}^{2}}}}\left( -x\sin {{\Phi }_{k}}+y\cos {{\Phi }_{k}} \right)\omega {{T}_{0}}m}}$, $\forall m,k$.
Therefore, the Doppler frequency of the scatterer
$\left(x, y\right)$ can be calculated as
${{D}}\left( x,y;{{\Phi }_{k}} \right)= \frac{1}{2\pi }\frac{\partial }{\partial \left( {{T}_{0}}m \right)}\text{angle}\left\{ {{y}_{m,k}}\left| _{\left( x,y \right)} \right. \right\}$,
from which (\ref{weliwjleikjr}) is derived.
\end{proof}
\vspace{0.5em}

Therefore, as indicated by Proposition \ref{proposition_sldkfjslkdj},
we can readily conclude that
the Doppler patterns of scatterers located at different $\left(x, y\right)$
positions are distinct,
which is a critical property for ensuring the uniqueness
of the mapping from the scatterer position to the Doppler pattern.

\vspace{0.5em}

\subsubsection{Unambiguous Imaging Region}
Applying the Nyquist-Shannon sampling theorem,
the unambiguous imaging region of the BI imaging
is presented below.
\vspace{0.5em}
\begin{proposition}\label{proposition_UIR}
(\emph{Unambiguous Imaging Region})
The unambiguous imaging region forms a circular area on the 2D plane,
given by
\begin{align}
\sqrt{x^{2}+y^{2}}\le \frac{\lambda }{4{{\omega }}{{T}_{0}}}\sqrt{1+\frac{H^{2}}{R^{2}}}.
\end{align}

\end{proposition}
\begin{proof}
Please refer to Appendix \ref{Appendix_proposition_UIR}.
\end{proof}
\vspace{0.5em}

A larger value of ${\lambda }{\sqrt{1+\left( {{{H}^{2}}}/{{{R}^{2}}} \right)}}$,
a shorter symbol interval $T_0$,
and a slower angular velocity $\omega$
collectively lead to
an expansion of the unambiguous imaging region,
which means that the proposed imaging scheme is capable of
imaging targets with large physical sizes
and preventing ambiguities in the reconstructed image.
This can serve as a practical guideline for parameter design.

\section{Imaging Algorithm for BI Imaging Scheme}\label{Section_sdlfjksldkfjlk}

Our proposed imaging algorithm consists of three main steps.
In Step 1,
we deal with the range term in (\ref{sdlsdkdkjci})
and utilize the range compression to enhance the SNR.
In particular,
after the range compression,
the signal model can be connected to an anisotropic CT model,
which is a very interesting finding.
In Step 2,
because there exists a one-to-one mapping
between the scatterer position and the angle-Doppler pattern
as unveiled in Proposition \ref{proposition_sldkfjslkdj},
we compute the angle-Doppler map
$g\left( D, \Phi \right)$ at each observation angle $\Phi$,
where $D$ denotes the Doppler frequency.
Finally, In Step 3,
we construct a mapping function
that maps the angle-Doppler map
onto the $\left(x,y\right) \in \mathcal{V}$ coordinate system
and achieves the imaging.

\subsection{Step 1: Range Compression}\label{Section_lskjflskdjfskldjf}

Since the bandwidth cannot provide extra information
for resolving scatterers within the ROI,
in this subsection,
the range compression is performed in the frequency domain to
improve the SNR.
Because
$c_{n,\ell}$,
$\xi$,
and
$\sqrt{R^2 + H^2}$
are known to the
BS\footnote{In both imaging scenarios,
because
the parameters of the transmitter, the receiver,
and
the distance between the ROI and the BS
are known to the BS,
the path loss $\xi$
can be readily determined by the radar equation.},
we can realize the range compression by (\ref{imaging_model_1}),
given at the top of the next page,
\begin{figure*}[!t]
\begin{align}\label{imaging_model_1}
{{T}_{m,k}}=\frac{1}{N}\sum\limits_{n}{{{e}^{j2\pi n\Delta f\frac{2\sqrt{{{R}^{2}}+{{H}^{2}}}}{c}}}\frac{{{S}_{n,m,k}}}{{{c}_{n,\ell }}\xi }} =\iint_{\mathcal{V}}{{{{\tilde{f}}}_{2D}}\left( x,y;{{\Phi }_{k}} \right)
\underbrace{{e}^{j\frac{4\pi }{\lambda }\frac{R}{\sqrt{{{R}^{2}}+{{H}^{2}}}}\left( -x\sin {{\Phi }_{k}}+y\cos {{\Phi }_{k}} \right)\omega {{T}_{0}}m}}_{\text{Doppler Term}}  dxdy}+{{{\tilde{Z}}}_{m,k}}, \ \forall m,k.
\end{align}
\normalsize
\hrulefill
\end{figure*}
where
${{{\tilde{Z}}}_{m,k}}=\frac{1}{N}\sum_{n}{{{e}^{j2\pi n\Delta f\frac{2\sqrt{{{R}^{2}}+{{H}^{2}}}}{c}}}\frac{{{Z}_{n,m,k}}}{{{c}_{n,\ell }}\xi }}$.

\vspace{0.5em}

An insightful observation from (\ref{imaging_model_1})
is that the signal model can be reformulated as an anisotropic
CT imaging model, which is formally summarized in the following remark.

\begin{remark}
(BI Imaging versus CT Imaging)
In first-generation CT systems,
after applying the FFT to the received attenuation data at the observation angle $\phi$,
the signal model of the frequency component $\varpi$
is given by \cite{CT_boooook}
\begin{align}\label{sdfsdfsldkfjslkjf}
{{S}_{CT}}\left( \varpi ,\phi  \right) =\iint_{\mathcal{V}}{{{f}_{CT}}\left( x,y \right){{e}^{-j2\pi \varpi \left( x\cos \phi +y\sin \phi  \right)}}dxdy},
\end{align}
where $f_{CT}\left( x,y \right) \in \mathbb{R}$ denotes the attenuation function of the target to be imaged in a CT system.
A comparison between the proposed BI imaging (\ref{imaging_model_1})
and the conventional CT model (\ref{sdfsdfsldkfjslkjf}) reveals both differences and connections.
The difference lies in the nature of the target functions.
In CT,
the attenuation function $f_{CT}\left( x,y \right)$ is isotropic.
In contrast,
the scattering function in radio imaging,
${{{\tilde{f}}}_{2D}}\left( x,y;{{\Phi }_{k}} \right)$,
is typically anisotropic
due to complex scattering phenomena
such as occlusion
and diffraction.
Despite this difference,
the two models
\emph{share identical exponential terms}
upon a simple coordinate transformation,
i.e.,
$s=-\frac{2}{\lambda }\frac{R}{\sqrt{{{R}^{2}}+{{H}^{2}}}}\omega {{T}_{0}}m$
and $\varphi ={{\Phi }_{k}}+\frac{\pi }{2}$.
Therefore,
the signal model of BI imaging can be regarded as an anisotropic CT signal model.
This anisotropy,
however,
causes the differences in algorithm design.
In standard CT imaging,
the
isotropic properties ensure
coherence among the received attenuation data.
Therefore,
some typical algorithms specifically
devised for CT imaging,
such as
the filtered back projection (FBP) algorithm \cite{CT_boooook},
apply coherent processing to reconstruct the image.
In contrast, the anisotropic nature
of BI imaging
hinders the coherence
of signals
across different observation angles.
As a result,
conventional CT algorithms cannot be employed in our model
and
only incoherent imaging methods are applicable,
as detailed in the following subsections.

\end{remark}

\subsection{Step 2: Generation of Angle-Doppler Map}\label{Section_IAA_super_resolution}

In this subsection, the angle-Doppler map is generated to facilitate the reconstruction in Step 3.

Physically,
the angle-Doppler map characterizes
the spectrum at Doppler frequency $D$
when the ROI is observed from angle $\Phi$.
Based on the above definition,
the angle-Doppler map can be readily defined as
\begin{align}\label{AD_Map}
\nonumber {{g}}\left( {{D}}, \Phi  \right)  \triangleq {}&  \iint_{\mathcal{V}}{{{{\tilde{f }}}_{2D}}\left( x,y;{{\Phi }} \right) \delta \left( D - {{D}_f}\left( x,y;{{\Phi }} \right) \right)dxdy}, \\
& \forall D \in \left[-\frac{1}{2T_0}, \frac{1}{2T_0}\right), \Phi \in \left[0, 2 \pi\right).
\end{align}
Our goal is to obtain ${{g}}\left( {{D}}, \Phi  \right)$
using ${{T}_{m,k}}$.

Next, we show how the super-resolution algorithm, IAA \cite{5417172},
is applied to generate the angle-Doppler map for our imaging problem.
Let
$\boldsymbol{b}\left( D \right)\in {{\mathbb{C}}^{\left(2M + 1\right)\times 1}}$,
$\boldsymbol{t}_k \in \mathbb{C}^{\left(2M + 1\right) \times 1}$,
and ${{\boldsymbol{z}}_{k}}\in \mathbb{C}^{\left(2M + 1\right) \times 1}$
denote
the vectors of
the exponential basis function,
the signal model,
and the noise
at the $k$-th azimuth angle, $\forall k$, respectively,
where
${{\left[ \boldsymbol{b}\left( D \right) \right]}_{m}}={{e}^{j{2\pi }D {{T}_{0}}m}}$,
$\forall m$,
$\left[\boldsymbol{t}_k\right]_m = {T}_{m,k}$
and $\left[\boldsymbol{z}_k\right]_m = \tilde{Z}_{m,k}$,
$\forall m,k$.
At the $k$-th azimuth angle,
(\ref{imaging_model_1}) can be rewritten as
a more compact vector form
\begin{align}
{{\boldsymbol{t}}_{k}}    =    \iint_{\mathcal{V}}{  {{{\tilde{f }}}_{2D}} \left( x,y;{{\Phi }_{k}} \right) \boldsymbol{b}\left( {{D}_f}\left( x,y; {{\Phi }_{k}} \right)\right)dxdy} +  {{\boldsymbol{z}}_{k}}, \ \forall k.
\end{align}

To facilitate practical implementation,
we discretize the Doppler frequency range $ D \in \left[-\frac{1}{2T_0}, \frac{1}{2T_0}\right)$ into
a uniform grid of $I$ points,
with $D_i = -\frac{1}{2T_0} + \frac{i}{IT_0}$, $\forall i = 0, \cdots, I - 1$,
denoting the frequency value at the $i$-th grid point.
The data covariance matrix of ${{\boldsymbol{t}}_{k}}$ in IAA algorithm is defined by
\begin{align}\label{dlfjlskdjflsk}
{{\boldsymbol{R}}_{k}}= \sum\limits_{i}{\boldsymbol{b}  \left( {{D_i}}\right){{\left| {{g}}\left( {{D_i}}, \Phi_k  \right) \right|}^{2}}{{\boldsymbol{b}}^{H}}   \left( {{D_i}} \right)}+\sigma _{z}^{2}{{\boldsymbol{I}}_{M}}, \ \forall k,
\end{align}
where
${{\boldsymbol{R}}_{k}} \in \mathbb{C}^{\left(2M + 1\right) \times \left(2M + 1\right)}$, $\forall k$,
and
${{g}}\left( {{D_i}}, \Phi_k  \right)  \in \mathbb{C}$, $\forall i, k$,
is the angle-Doppler map
of the Doppler frequency ${D_i}$,
at the azimuth angle $\Phi_k$.

Next, we apply the IAA algorithm at each azimuth angle,
$\Phi_k$, $\forall k$,
to generate the angle-Doppler map.
In IAA algorithm,
when estimating the signal at the given
Doppler frequency ${D_i}$,
all other signal components are treated as interference.
We define the interference-plus-noise covariance matrix of the Doppler frequency ${D_i}$
at the azimuth angle $\Phi_k$ as
${{\boldsymbol{Q}}_{k}}\left( {D_i} \right)={{\boldsymbol{R}}_{k}}-{{\boldsymbol{b}}^{H}}\left( {D_i} \right){{\left| {{g}}\left( {D_i}, \Phi_k  \right)  \right|}^{2}}\boldsymbol{b}\left( {D_i} \right)$, $\forall i,k$.
The IAA formulates the following weighted least-squares problem to estimate
$g \left( {D_i}, \Phi_k \right)$
\begin{align}\label{problemwsdfkljhl}
\underset{g \left( {D_i}, \Phi_k \right)}{\mathop{\min }}\,\left\| {{\boldsymbol{y}}_{k}}-{{g}}\left( {D_i}, \Phi_k  \right) \boldsymbol{b}\left( {D_i} \right) \right\|_{\boldsymbol{Q}_{k}^{-1}\left( {D_i} \right)}^{2}, \ \forall i,k.
\end{align}
Then, an iterative estimation algorithm is devised for this problem.
For the azimuth angle $\Phi_k$,
at the $j$-th iteration, the estimator for the Doppler frequency ${D_i}$ is
\begin{align}\label{welijhlkerjhl53kgrdflm}
{{{{g }}}_{j+1}}\left( {D_i}, \Phi_k \right)=\frac{{{\boldsymbol{b}}^{H}}\left( {D_i} \right)\boldsymbol{R}_{j, k}^{-1}{{\boldsymbol{y}}_{k}}}{{{\boldsymbol{b}}^{H}}\left( {D_i} \right)\boldsymbol{R}_{j, k}^{-1}\boldsymbol{b}\left( {D_i} \right)}, \ \forall j,k,i,
\end{align}
where
$\boldsymbol{R}_{j,k}$
is obtained by substituting
${{g}}_j\left( {D_i}, \Phi_k  \right) $
into (\ref{dlfjlskdjflsk})
and $\boldsymbol{R}_{j,k}$ is initialized as $\boldsymbol{I}_M$.

Upon convergence,
the estimate of
the angle-Doppler map
at
the Doppler frequency ${D_i}$ of the azimuth angle $\Phi_k$ can be expressed as
\begin{align}\label{dlkfjlskj}
\nonumber  &{\hat{g}}\left( {D_i}, \Phi_k  \right)  \approx  \iint_{\mathcal{V}}{{{{\tilde{f }}}_{2D}}\left( x,y;{{\Phi }_{k}} \right)}\\
  & \ \ \ \ \ \ \ \ \ \ \ \times \chi \left( {D_i} - {{D_f}}\left( x,y;{{\Phi }_{k}} \right) \right)dxdy+{{z}\left( {D_i}, \Phi_k  \right)}, \, \forall i,k,
\end{align}
where $\chi\left(\cdot\right)$
is the ambiguity function of the IAA algorithm with $\chi\left(0\right) = 1$,
and ${{z}\left( {D_i}, \Phi_k  \right)}$ denotes the noise term of the Doppler frequency ${D_i}$ at the azimuth angle $\Phi_k$.

\subsection{Step 3: Doppler Association}\label{Section_EHT}

\begin{algorithm}[!t]
\caption{Proposed IAA-DA Algorithm}\label{Algorithm1}
\hspace*{0.02in}{\bf Input:}
the received signal,
${{s}_{n,\ell }}$, $\forall n, \ell$,
the small-angle parameter, $\delta_{SA}$,
and the number of the
azimuth angle centers, $K$.\\
\hspace*{0.02in}{\bf Output:}
the energy map of the ROI, $G\left( x,y \right)$, $\forall        \left(x, y\right)     \in     \mathcal{V}$.

\begin{algorithmic}[1]
\STATE
Determine the angular correlation range by (\ref{sdlfkjsldkfj})
and
reshape the received signal to ${{S}_{n,m,k}}$ based on (\ref{sdlsdkdkjci}), $\forall n,m,k$.

\STATE
Perform range compression on ${{S}_{n,m,k}}$ according to
(\ref{imaging_model_1})
to yield ${{T}_{m,k}}$, $\forall m,k$.

\FOR{$k=0$ to $K-1$}
\STATE
Initialize $\boldsymbol{R}_{0,k}$ as $\boldsymbol{I}_M$
and employ (\ref{welijhlkerjhl53kgrdflm})
to solve the problem (\ref{problemwsdfkljhl}).

\STATE
Obtain ${\hat{g}}\left( {D_i}, \Phi_k  \right)$ after the convergence of (\ref{welijhlkerjhl53kgrdflm}).

\ENDFOR

\STATE

Use (\ref{dlijfldkjf})
to generate the energy map
of the angle-Doppler map,
${\hat{g}}^{+}\left( {D_i}, \Phi_k  \right)$.

\STATE
Perform (\ref{sdlfkjsdlkfjskldj})
to achieve the Doppler association
and reconstruct
the energy map of the ROI, $G\left( x,y \right)$, $\forall        \left(x, y\right)     \in     \mathcal{V}$.

\end{algorithmic}
\end{algorithm}

Proposition \ref{proposition_sldkfjslkdj}
reveals
a \emph{one-to-one mapping} between the spatial location $\left(x, y\right)$ and the Doppler pattern $D_f  \left( x,y;{{\Phi }_{k}} \right)$.
Thus,
the angle-Doppler map in (\ref{dlkfjlskj}) can be uniquely projected onto the 2D spatial coordinate system via a Doppler association process.
Before proceeding,
it should be noted that the anisotropy of the scattering function
introduces randomness
in both amplitude and, especially, phase,
thereby hindering the coherent integration.
Due to their randomness,
we can only apply the incoherent integration
for Doppler association.
One of the methods to achieve this is to compute the magnitude of ${\hat{g}}\left( {D_i}, \Phi_k  \right)$
to generate an energy map
\begin{align}\label{dlijfldkjf}
{\hat{g}}^{+}\left( {D_i}, \Phi_k  \right)=  \left| {\hat{g}}\left( {D_i}, \Phi_k  \right) \right|, \ \forall i,k.
\end{align}

Next, we introduce our Doppler association algorithm.
Our Doppler association is an integral transform
that maps sinusoidally modulated Doppler frequencies
onto the $\left(x,y\right) \in \mathcal{V}$ coordinate system.
Specifically,
imaging result, $G\left( x,y \right) \in \mathbb{R}$, is given by
\begin{align}\label{sdlfkjsdlkfjskldj}
 G\left( x,y \right)    =     \frac{1}{K}\sum\limits_{k=0}^{K-1}    {\hat{g}^{+}    \left( Q\left( x,y;{{\Phi }_{k}} \right)     , \Phi_k \right)}, \ \forall        \left(x, y\right)     \in     \mathcal{V},
\end{align}
where
${Q}\left( x,y;{{\Phi }_{k}} \right)=\left[ \frac{D_f\left( x,y;{{\Phi }_{k}} \right)}{{{f}_{\Delta }}} \right]{{f}_{\Delta }}$
is the quantized Doppler frequency,
and ${{f}_{\Delta }} = \frac{1}{IT_0}$ is the length of the Doppler grid.

Algorithm \ref{Algorithm1}
summarizes the proposed IAA-DA algorithm for BI imaging schemes.

\section{Resolution Analysis of the Imaging Algorithm}\label{Section_Algorithm_Performance}

In this section,
we characterize the imaging resolution
for the proposed BI imaging.
The point spread function (PSF) describes how an imaging system responds to a point source.
Moreover,
the shape of the PSF directly defines the theoretical resolution limit of the system.
Therefore, in this section,
we first focus on the PSF for our algorithm,
and then derive the imaging resolution.

\subsection{Point Spread Function}

However,
(\ref{dlijfldkjf})
is typically a nonlinear
operation.
Referring to the derivation in Appendix \ref{Appendix_sdflkjslfkjsldkjf}, in a noise-free environment,
(\ref{dlijfldkjf}) can be approximated
to a linear expression
\begin{align}\label{sdflkjslfkjsldkjf}
\nonumber & {\hat{g}}^{+}\left( {D_i}, \Phi_k  \right) \approx \\
&\iint{  {{{\tilde{f }}}_{2D}^{+}}  \left( x,y;{{\Phi }_{k}} \right)\chi^{+}   \left( {D_i}  -  {{D_f}}\left( x,y;{{\Phi }_{k}} \right) \right)dxdy}, \ \forall k, i,
\end{align}
where
$\tilde{f }_{2D}^{+}\left( x,y;{{\Phi }_{k}} \right) \in \mathbb{R}$
is the equivalent scattering function,
${{\chi }^{+}}\left( D \right)$ is the equivalent ambiguity function
defined as $\left| \chi \left( D \right) \right|$
for $D\in \left( -{{D}_{\text{ML}}},{{D}_{\text{ML}}} \right)$
and zero otherwise,
and ${{D}_{\text{ML}}}$
is the mainlobe range of $\chi \left( D \right)$.

Next, we derive the PSF based on (\ref{sdlfkjsdlkfjskldj}) and (\ref{sdflkjslfkjsldkjf}).
For simplicity, we consider a single scatterer located at the origin $\left(0,0\right)$.
It is important to note that the subsequent derivation and analysis hold without loss of generality
for a scatterer at any arbitrary position.
We have
\begin{align}\label{sdfjskldkf}
{{\tilde{f}}_{2D}^{+}}\left( x,y;{{\Phi }} \right)={{\tilde{f}}_{2D}^{+}}\left( {{\Phi }} \right)\delta \left( x,y \right), \ \forall        \left(x, y\right)     \in     \mathcal{V},
\end{align}
where $\delta \left( x, y \right)$ is the Dirac delta function.
Then, the PSF is provided in the following proposition.
\vspace{0.5em}
\begin{proposition}\label{proposition_PSF}
(\emph{Point Spread Function of the Imaging Algorithm})
For a single scatterer located at $\left(0,0\right)$,
the PSF is
\begin{align}\label{jsflkdjlfskjlkd}
\nonumber &\text{PSF}\left( x,y \right)\\
& =  \frac{1}{2\pi } \int_{0}^{2\pi }{\tilde{f }_{2D}^{+}\left( \Phi -{{\varphi }_{x,y}} \right){{\chi }^{+}}\left(\beta^{-1} {{\sqrt{{{x}^{2}}+{{y}^{2}}}}}\cos \Phi  \right)d\Phi },
\end{align}
where
${{\varphi }_{x,y}}=\arccos \frac{y}{\sqrt{{{x}^{2}}+{{y}^{2}}}}$
and
$\beta = \frac{\lambda {}}{2\omega }\sqrt{1+\frac{{{H}^{2}}}{{{R}^{2}}}}$.
In particular,
the value of the PSF at the origin is
\begin{align}\label{Dfsdlfksl}
\text{PSF}\left( 0,0 \right)=\frac{1}{2\pi }\int_{0}^{2\pi }{\tilde{f }_{2D}^{+}\left( \Phi  \right)d\Phi }.
\end{align}
\end{proposition}
\begin{proof}
It is straightforward to prove (\ref{jsflkdjlfskjlkd}) and (\ref{Dfsdlfksl}) using
(\ref{sdlfkjsdlkfjskldj})
and
(\ref{sdflkjslfkjsldkjf}),
and is
therefore omitted for simplicity.
\end{proof}
\vspace{0.5em}

\begin{remark}\label{skfjsdlkfjsdkdfj}
Based on (\ref{Dfsdlfksl}),
we can conclude that
in the single-scatterer case,
the proposed algorithm aggregates the backscattered signal energy from all observation angles
at the scatterer position.
For general cases with multiple scatterers,
$G\left( x,y \right)$
can also be regarded as an approximate energy map.
\end{remark}

\subsection{Imaging Resolution}

First,
we present the following resolution concept for our proposed scheme.
\vspace{0.2em}
\begin{define}\label{Definition_resolution}
(\emph{Resolution})
Let
$\left( \zeta \left( \alpha  \right)\sin \alpha ,\zeta\left( \alpha  \right)\cos \alpha  \right)$
denote the $x$-$y$ coordinates of
the $-3\text{dB}$ contour of
the normalized
$\frac{\text{PSF}\left( x,y \right)}{\text{PSF}\left( 0,0 \right)}$,
where
$\zeta\left( \alpha  \right)\in {{\mathbb{R}}_{\ge 0}}$
and
$\alpha \in \left[0, 2\pi\right]$.
The resolution is then defined as
\begin{align}
{{\delta }_{\text{IAA-DA}}}=\underset{\alpha }{\mathop{\max }}\,\left[\zeta\left( \alpha  \right) + \zeta\left( \alpha + \pi  \right)\right].
\end{align}
\end{define}
\vspace{0.2em}

Note that the defined resolution represents \emph{the worst resolution} achievable over all directions.
We provide the following imaging resolution
based on Definition \ref{Definition_resolution}.

\vspace{0.2em}
\begin{proposition}\label{proposition_resolution}
(\emph{Resolution of the Imaging Algorithm})
The resolution of the proposed IAA-DA algorithm,
denoted by $\delta_{\text{IAA-DA}}$,
is given by
\begin{align}\label{dsyysdsdsfkjsddfsdfslkfj}
{{\delta }_{\text{IAA-DA}}}={{\kappa }_{I}}{{\kappa }_{S}}\delta _{\text{benchmark}} = {{\kappa }_{I}}{{\kappa }_{S}}{{C}_{R}}\frac{\lambda }{4{{\delta }_{SA}}}\sqrt{1+\frac{{{H}^{2}}}{{{R}^{2}}}},
\end{align}
where
$\delta_{\text{benchmark}} = {{C}_{R}}\frac{\lambda }{4{{\delta }_{SA}}}\sqrt{1+\frac{{{H}^{2}}}{{{R}^{2}}}}$
is the resolution benchmark;
$C_R \approx 1.29$ is a constant;
$\kappa_I \ge 1$
is the degradation coefficient
due to the anisotropic property of the scattering function,
and $\kappa_I = 1$
when the scattering function is isotropic;
$\kappa_S \le 1$
is the resolution gain
from the super-resolution algorithm,
and
$\kappa_S = 1$
when applying the DFT algorithm to generate the angle-Doppler map.

\end{proposition}
\begin{proof}
Please refer to Appendix \ref{Appendix_proposition_resolution}.
\end{proof}
\vspace{0.2em}

In particular,
the resolution of the DFT-DA imaging algorithm, which is
obtained by replacing the IAA estimator with a standard DFT
method in the IAA-DA algorithm,
can be expressed as
${{\delta }_{\text{DFT-DA}}}={{\kappa }_{I}}\delta _{\text{benchmark}}$.

\vspace{0.5em}
\begin{remark}
(\emph{Equivalent Synthesized Bandwidth})
In conventional SAR and ISAR imaging systems,
2D reconstruction relies on
signal bandwidth for range resolution
and Doppler processing for cross-range resolution.
In contrast,
the proposed scheme
utilizes
Doppler information
to resolve both imaging dimensions.
Thus,
an equivalent synthesized bandwidth
is achieved through Doppler processing.
Denoted by
$B_{\text{IAA-DA}}^{\text{ESB}}$,
the equivalent bandwidth
is given by
\begin{align}
B_{\text{IAA-DA}}^{\text{ESB}}=\frac{c}{2{{\delta }_{\text{IAA-DA}}}}=\frac{1}{{{\kappa }_{I}}{{\kappa }_{S}}{{C}_{R}}}\frac{2{{\delta }_{SA}}{{f}_{0}}}{\sqrt{1+\frac{{{H}^{2}}}{{{R}^{2}}}}}.
\end{align}

\end{remark}
\vspace{0.5em}

For instance,
considering a BI imaging scenario
with
carrier frequency 30 GHz,
the height-radius ratio
${H}/{R} = 0.1$,
the small-angle parameter
$\delta_{SA} = 3{}^\circ  \frac{\pi}{180{}^\circ} \, \text{rad}$,
the degradation coefficient $\kappa_I = 1.3$,
and the resolution gain $\kappa_S = 0.5$,
we can calculate that
the resolution benchmark
is
$\delta_{\text{benchmark}} = 6.19 \, \text{cm}$
and the imaging resolution
is
${{\delta }_{\text{IAA-DA}}} = 4.02 \, \text{cm}$,
achieving centimeter-level imaging resolution.
Furthermore,
the equivalent synthesized bandwidth is $B_{\text{IAA-DA}}^{\text{ESB}}=3.73\, \text{GHz}$.

\vspace{0.5em}
\begin{remark}
(\emph{Resolution Analysis})
A key observation from Proposition \ref{proposition_resolution} is that the imaging resolution of the proposed algorithm
is independent of the target-BS distance.
This is because the BI imaging scheme relies solely on Doppler frequency.
Additionally,
as the resolution depends on the wavelength $\lambda$,
our proposed BI scheme is particularly well-suited for implementation in
systems with shorter wavelengths,
such as
mmWave communication systems.

\end{remark}

\section{Extension of the IAA-DA Algorithm to General Cases with Large-Sized ROI}\label{Section_large_size}

The above proposed IAA-DA algorithm and its analysis
are based on the small-sized ROI constraint, as specified in (\ref{welikjrwlekjr}).
In this section,
we extend the proposed algorithm
to more general imaging scenarios
with a large-sized ROI.
In such cases, targets within the ROI may occupy different range bins.
For instance, in a communication system with a bandwidth of 200 MHz,
i.e., corresponding to a range resolution of $0.75\text{m}$,
and an ROI with a maximum size of $2\text{m}$,
the target would occupy about three range bins.
This implies that the Doppler information contained within a single range bin is insufficient to represent all the scattering characteristics of targets.
To address this,
we propose an algorithm that effectively extends the IAA-DA approach to ROIs of arbitrary size. Specifically, we first perform imaging independently on each subcarrier and then synthesize the resulting images. The detailed procedure is described as follows.

Without the approximation given in (\ref{welidjfsdlkhgjlk}),
we omit the detailed derivation and directly present
the approximated signal model for large-sized ROI cases
in (\ref{welijlwkjnlgkn}), which is shown at the top of the next page.
\begin{figure*}[!t]
\begin{align}\label{welijlwkjnlgkn}
\nonumber  {{S}_{n,m,k}}& \approx {{c}_{n,\ell }}\xi \iint_{\mathcal{V}}{{{{\tilde{f}}}_{2D}}\left( x,y;n,{{\Phi }_{k}} \right){{e}^{j2\pi \frac{R}{\sqrt{{{R}^{2}}+{{H}^{2}}}}\left( \frac{2}{\lambda }+\frac{n\Delta f}{c} \right)\left( -x\sin {{\Phi }_{k}}+y\cos {{\Phi }_{k}} \right)\omega {{T}_{0}}m}}dxdy}+{{Z}_{n,m,k}} \\
 & \approx {{c}_{n,\ell }}\xi \iint_{\mathcal{V}}{{{{\tilde{f}}}_{2D}}\left( x,y;n,{{\Phi }_{k}} \right){{e}^{j\frac{4\pi }{\lambda }\frac{R}{\sqrt{{{R}^{2}}+{{H}^{2}}}}\left( -x\sin {{\Phi }_{k}}+y\cos {{\Phi }_{k}} \right)\omega {{T}_{0}}m}}dxdy}+{{Z}_{n,m,k}}, \ \forall n,m,k.
\end{align}
\normalsize
\hrulefill
\end{figure*}
In (\ref{welijlwkjnlgkn}),
we omit the term ${n  \Delta f}/{c}$ because ${2}/{\lambda }={2{{f}_{0}}}/{c}$
is much greater than ${n\Delta f}/{c}$,
and
${{{\tilde{f}}}_{2D}}\left( x,y;n,{{\Phi }_{k}} \right)$
denotes the scattering function of the $n$-th subcarrier, given by
\begin{align}\label{weligsfdgdfgjlwkjnlgkn}
\nonumber & {{{\tilde{f}}}_{2D}}\left( x,y;n,{{\Phi }_{k}} \right) =\\
\nonumber &{{e}^{-j2\pi \left( \frac{2}{\lambda }+\frac{n\Delta f}{c} \right)\sqrt{{{R}^{2}}+{{H}^{2}}}}}{{e}^{j2\pi \left( \frac{2}{\lambda }+\frac{n\Delta f}{c} \right)\frac{R}{\sqrt{{{R}^{2}}+{{H}^{2}}}}\left( x\cos {{\Phi }_{k}}+y\sin {{\Phi }_{k}} \right)}}\\
&\times \int_{\mathcal{V}}{{{f}_{3D}}\left( x,y,z;{{\phi }_{\ell }} \right){{e}^{j2\pi \left( \frac{2}{\lambda }+\frac{n\Delta f}{c} \right)\frac{zH}{\sqrt{{{R}^{2}}+{{H}^{2}}}}}}dz}, \  \forall n,k.
\end{align}
It can be seen from (\ref{welijlwkjnlgkn})
that for a specific subcarrier $n_1$,
its signal model ${{S}_{n_1,m,k}}$
is exactly the same as ${{T}_{m,k}}$,
as shown in (\ref{imaging_model_1}).

As received signal at each subcarrier is narrowband and naturally satisfies condition (\ref{welikjrwlekjr}),
we can independently
apply the IAA-DA algorithm to each subcarrier.
Let $G_n \left(x, y\right)\in\mathbb{R}$ denote the resulting image using the signal at the $n$-th subcarrier.
The synthesized image, denoted by $G_{\text{Syn}} \left(x, y\right)\in\mathbb{R}$, is obtained by
\begin{align}
G_{\text{Syn}} \left(x, y\right) = \frac{1}{N} \sum\limits_{n=0}^{N - 1}G_n \left(x, y\right), \ \forall        \left(x, y\right)     \in     \mathcal{V}.
\end{align}
Therefore,
$G_{\text{Syn}} \left(x, y\right)$
is the resulting image that synthesizes the images over all subcarriers.

\begin{figure*}[!t]
\centering
\subfigure[]
{\includegraphics[height=1.37in]{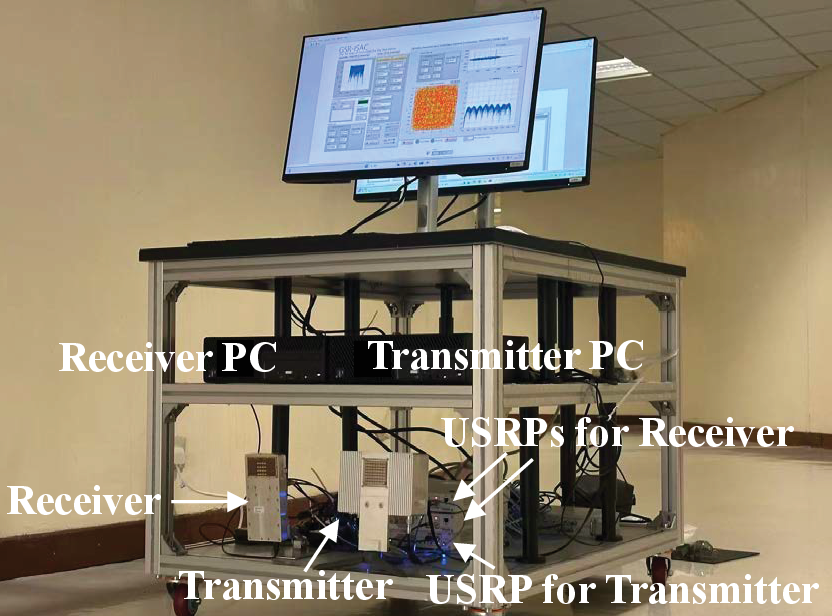}} \ \ \ \
\subfigure[]
{\includegraphics[height=1.37in]{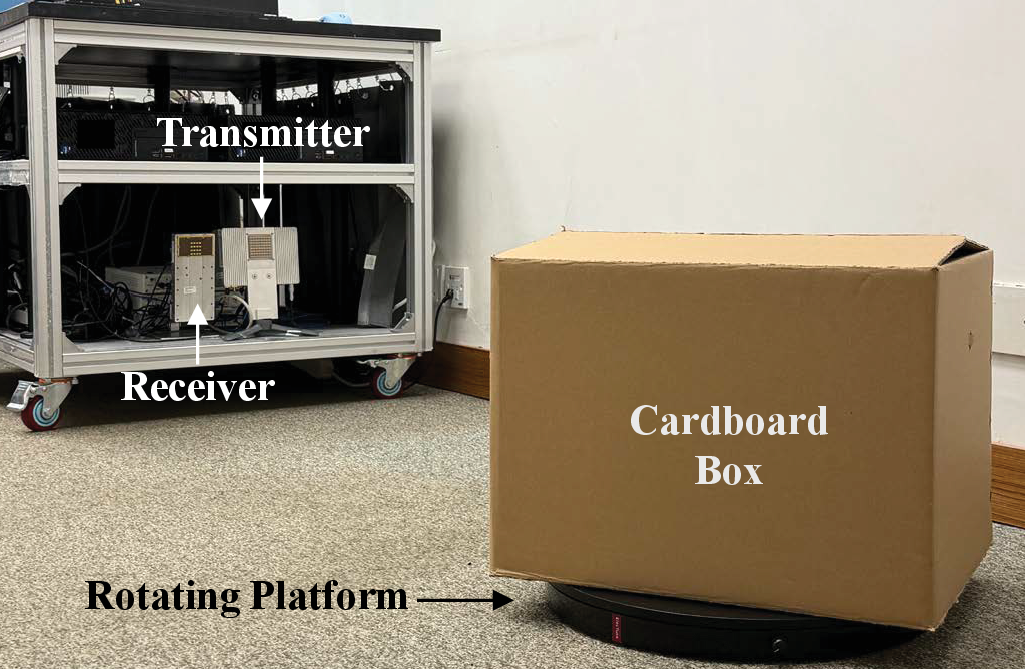}} \ \ \ \
\subfigure[]
{\includegraphics[height=1.37in]{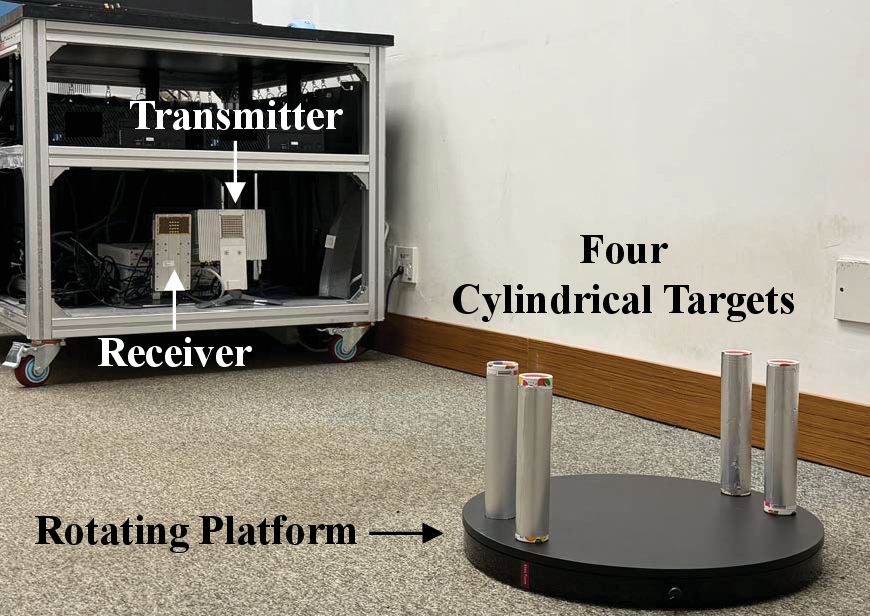}}
\caption{Experiment scenarios
with
(a) hardware,
(b) a box
with the size of
$0.35\text{m} \times 0.41\text{m} \times 0.33\text{m}$
(along $x$, $y$, $z$ axes),
and
(c) four cylindrical targets
located at
$\left(7.5, 14.5\right)\text{cm}$,
$\left(14.5, 9.5\right)\text{cm}$,
$\left(-12, -11\right)\text{cm}$,
and
$\left(-5.5, -15\right)\text{cm}$.}
\label{Fig_exp_box_results}
\end{figure*}

\begin{figure*}[!t]
\centering
\subfigure[]
{\includegraphics[width=1.65in]{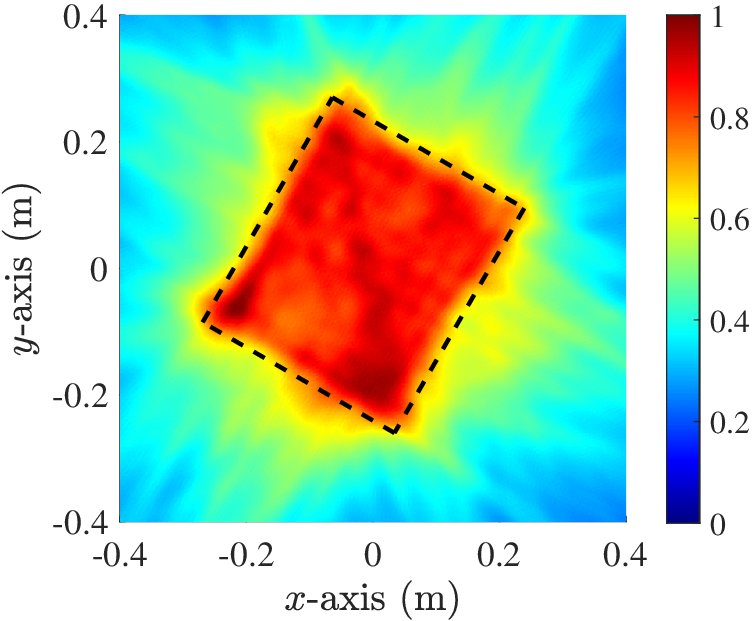}} \ \
\subfigure[]
{\includegraphics[width=1.65in]{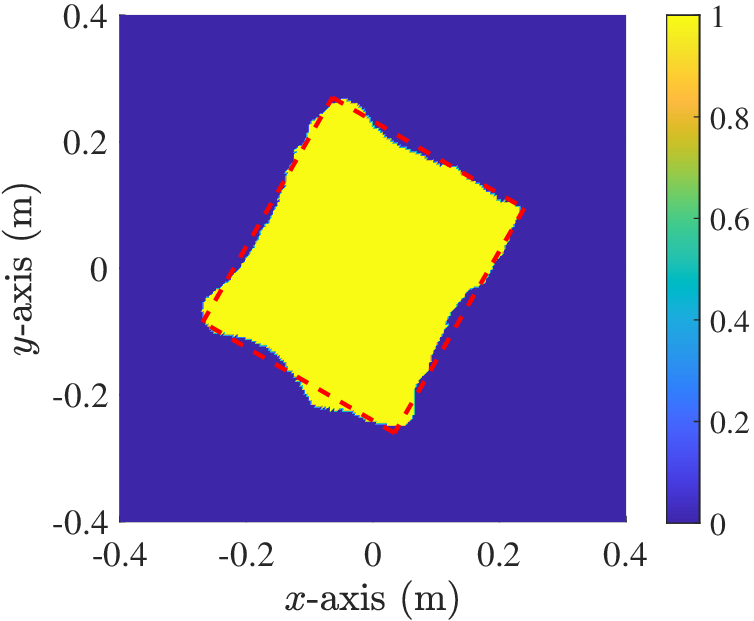}} \ \
\subfigure[]
{\includegraphics[width=1.65in]{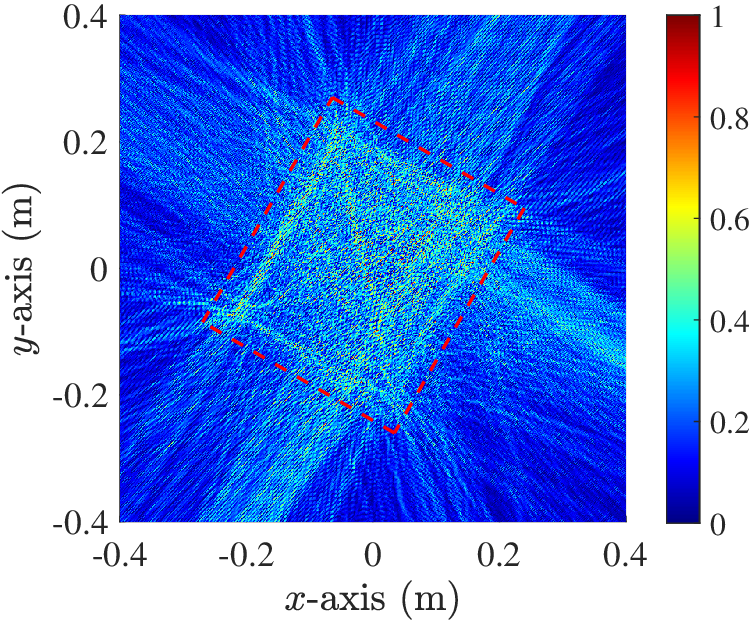}} \ \
\subfigure[]
{\includegraphics[width=1.65in]{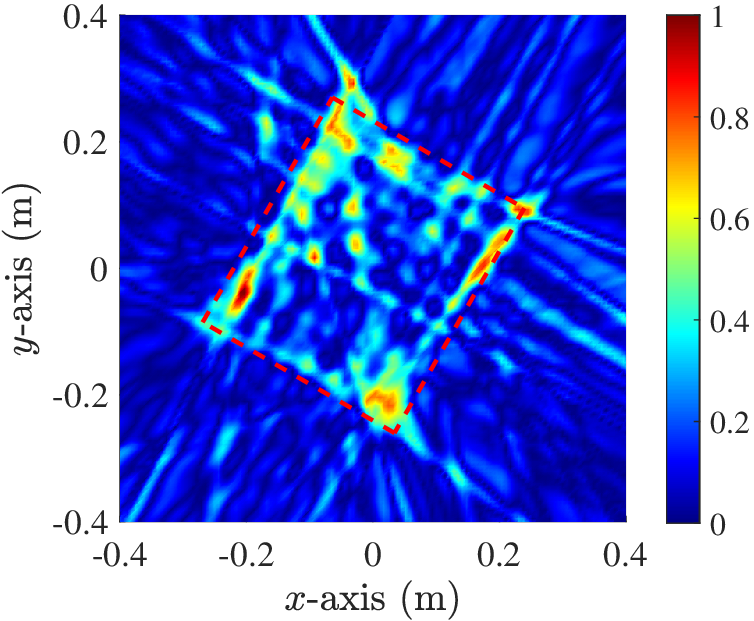}}
\caption{Imaging results for the box target.
(a)
Result of the proposed IAA-DA algorithm.
(b)
Detection result of the reconstructed image.
(c)
Result of the SRDI algorithm \cite{10475383}.
(d)
Result of the FBP algorithm \cite{3135465484121}.}
\label{Fig_exp_box_results1}
\end{figure*}

\section{Real-World Experiment Results}\label{Section_simulation_experiment}

In this section, real-world experiment results will be presented to verify
imaging algorithm in Section \ref{Section_sdlfjksldkfjlk}
and analysis in Sections \ref{Section_Algorithm_Performance}.
Specifically,
real-world experiments are conducted based on our mmWave platform,
as shown in Fig. \ref{Fig_exp_box_results} (a).
In such a platform,
the transmit antenna and the receive antenna are separated by $0.15\text{m}$,
such that self-interference can be well mitigated.
For the transmit signals,
we set the carrier frequency as $29 \, \text{GHz}$ at the mmWave band,
the number of sub-carriers as $N=1024$,
the sub-carrier spacing as $\Delta f=97.6 \, \text{kHz}$,
and the time interval between adjacent OFDM symbols used for imaging as $T_0 = 5 \, \text{ms}$.
Therefore,
the signal wavelength is $\lambda=1.03 \, \text{cm}$
and
the total signal bandwidth is 100 MHz,
which implies that the range resolution is $1.5 \, \text{m}$
and cannot provide any resolution information for imaging.

Moreover,
in this section,
we will only provide the experiment results for ISAR imaging.
As shown in Figs. \ref{Fig_exp_box_results} (b) and (c),
the transmit and receive antennas are static,
while the targets are placed on a rotating platform.
Specifically,
the spin time of the targets over 360 degrees is $20 \, \text{s}$,
which suggests that
the angular velocity
is $\omega = \frac{2 \pi}{20}  \,  {\text{rad}}/{\text{s}}$
and the number of OFDM symbols
for imaging is $L = \frac{2\pi}{\omega T_0} = 4000$.
Two scenarios are considered for the imaging experiments.
The first utilizes a cardboard box of
$0.35\text{m} \times 0.41\text{m} \times 0.33\text{m}$
as the target, as shown in Fig. \ref{Fig_exp_box_results} (b).
The second scenario,
as shown in Fig. \ref{Fig_exp_box_results} (c),
employs four cylindrical targets
covered with aluminized paper to increase the backscattered
power.
Additionally,
we set
the number of the
azimuth angle centers,
as defined in Section \ref{Section_Signal_Model_Approximation},
to
$K = 400$
and set the small-angle parameter to
$\delta_{SA} = 2{}^\circ  \frac{\pi}{180{}^\circ} \, \text{rad}$.
Given the above system parameters,
the benchmark resolution
is calculated as
${{\delta }_{\text{benchmark}}}=9.56\text{cm}$
according to Proposition \ref{proposition_resolution}.
By leveraging the super-resolution algorithm,
the total resolution can be further enhanced to achieve superior resolution,
i.e., ${{\kappa }_{I}}{{\kappa }_{S}} < 1$.

\begin{figure*}[!t]
\centering
\subfigure[]
{\includegraphics[width=1.65in]{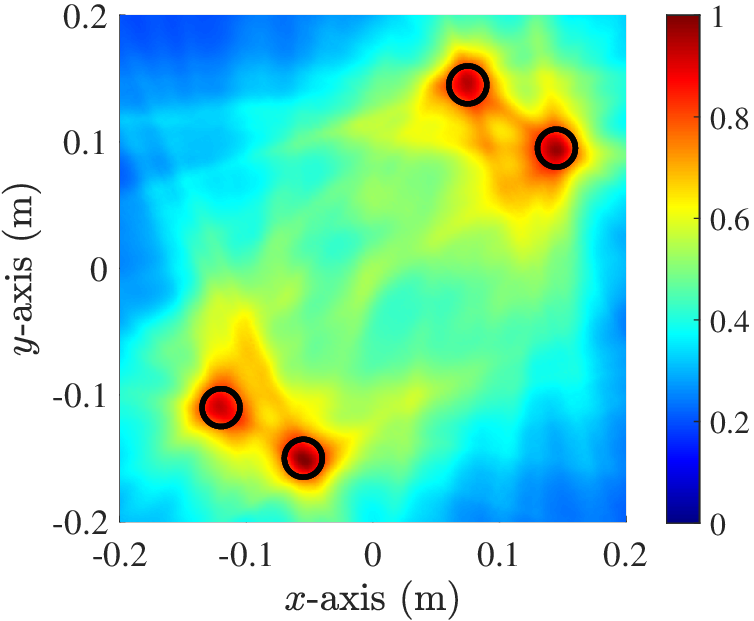}} \ \
\subfigure[]
{\includegraphics[width=1.65in]{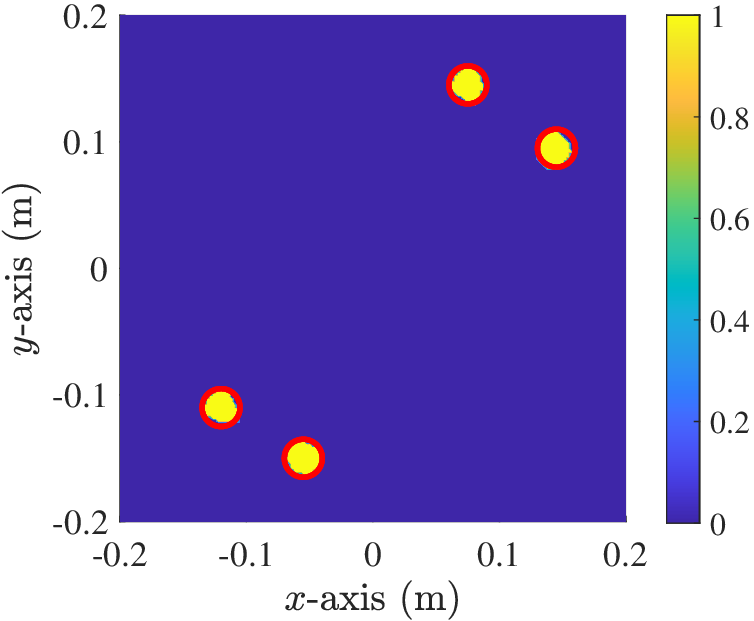}} \ \
\subfigure[]
{\includegraphics[width=1.65in]{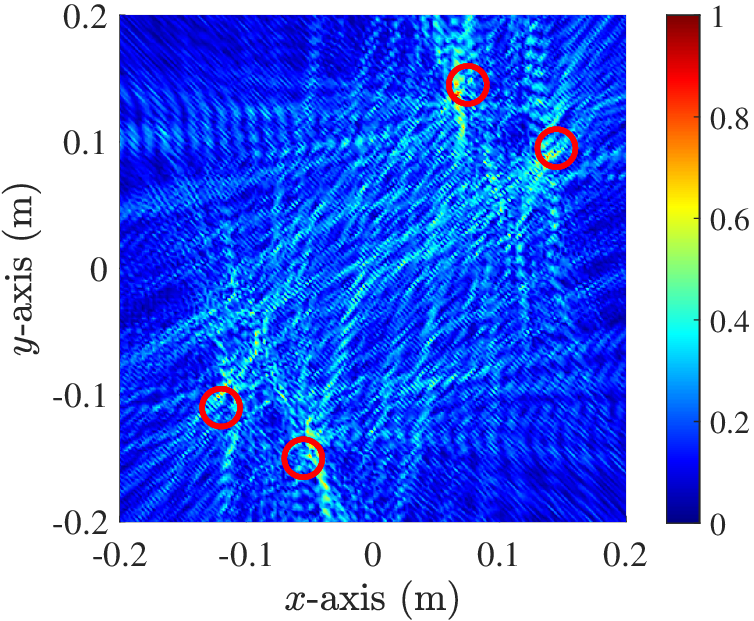}} \ \
\subfigure[]
{\includegraphics[width=1.65in]{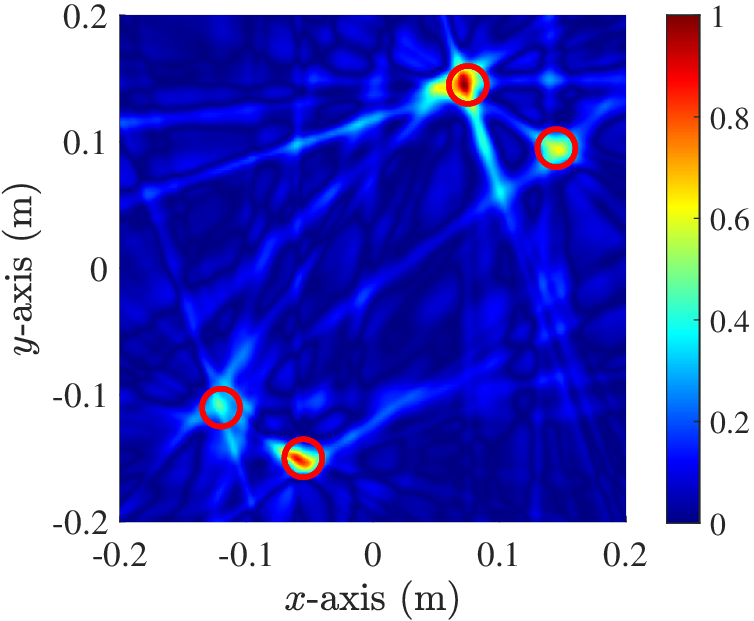}}
\caption{Imaging results for four cylindrical targets.
(a)
Result of the proposed IAA-DA algorithm.
(b)
Detection result of the reconstructed image.
(c)
Result of the SRDI algorithm \cite{10475383}.
(d)
Result of the FBP algorithm \cite{3135465484121}.}
\label{Fig_exp_box_results2}
\end{figure*}

\begin{figure*}[!t]
\centering
\subfigure[]
{\includegraphics[width=1.35in]{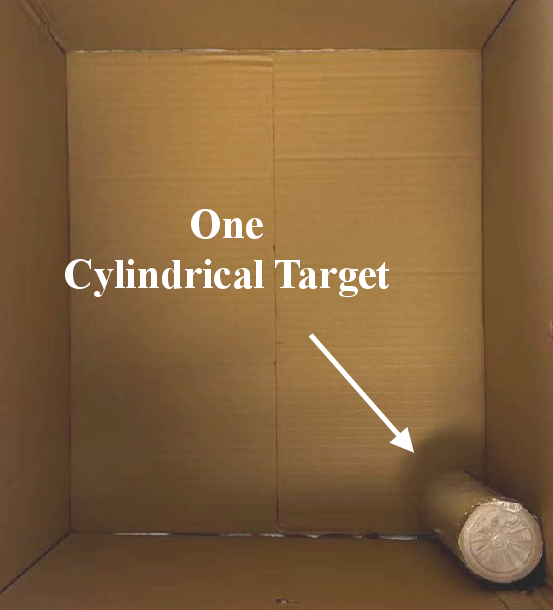}} \ \ \
\subfigure[]
{\includegraphics[width=1.65in]{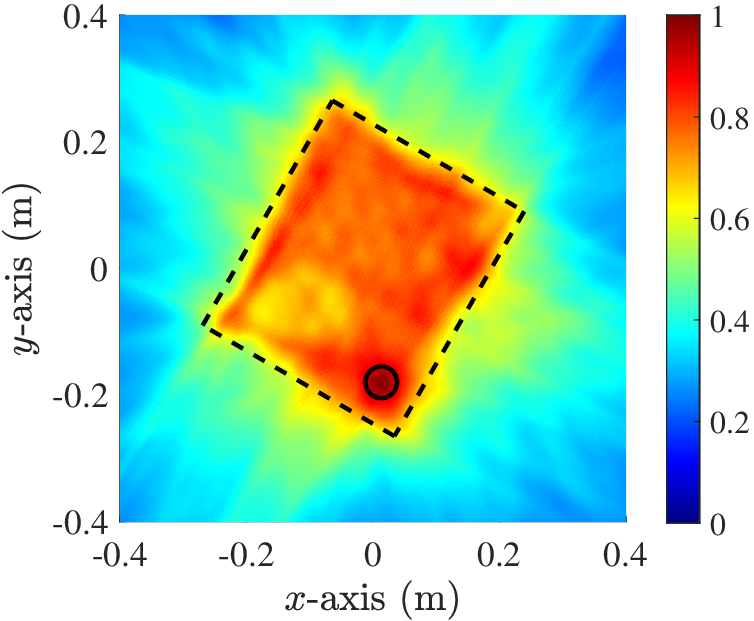}} \ \ \
\subfigure[]
{\includegraphics[width=1.35in]{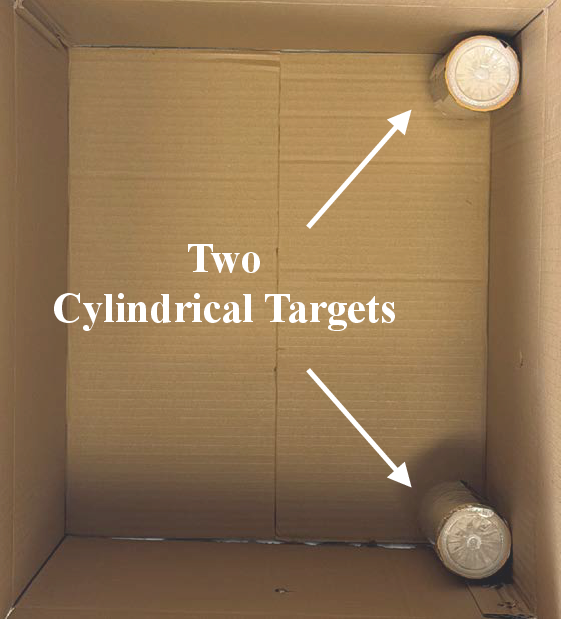}} \ \ \
\subfigure[]
{\includegraphics[width=1.65in]{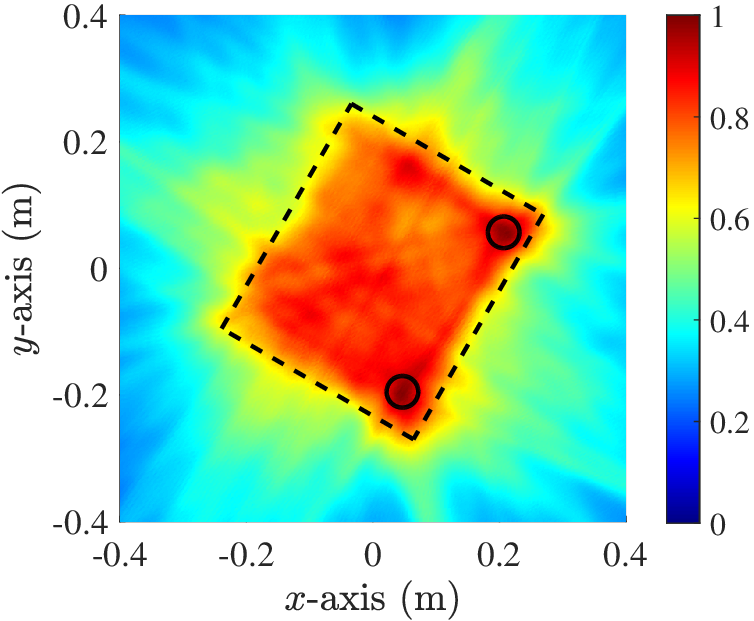}}
\caption{Results of the penetration imaging.
(a) Scenario 1 with one cylindrical target within the box.
(b) Result of the proposed IAA-DA algorithm for scenario 1.
(c) Scenario 2 with two cylindrical target within the box.
(d) Result of the proposed IAA-DA algorithm for scenario 2.}
\label{Fig_Imaging_Result_rectangle_901sdfds80270360}
\end{figure*}

Fig. \ref{Fig_exp_box_results1}
shows the experiment results
for a single box target.
The
image reconstructed by the proposed
IAA-DA algorithm is shown in
Fig. \ref{Fig_exp_box_results1} (a),
exhibiting sharply defined
edges and high overall quality.
After applying a
$-3 \, \text{dB}$ threshold
to the reconstructed image,
the resulting detection map in
Fig. \ref{Fig_exp_box_results1} (b)
demonstrates a high degree of fidelity to the physical dimensions of the box.
To highlight the superiority of the proposed algorithm, two comparative methods are also evaluated.
The first one is the SRDI algorithm \cite{10475383},
where the target is characterized by an isotropic scattering function.
As shown in
Fig. \ref{Fig_exp_box_results1} (c),
since physical targets typically exhibit anisotropic properties, this model mismatch causes the scatterer energy to smear across the entire image rather than concentrate at the true positions.
The second one is the FBP algorithm \cite{3135465484121},
specifically devised for CT imaging,
which also models an isotropic function for targets.
The FBP result
in
Fig. \ref{Fig_exp_box_results1} (d)
exhibits high sidelobes outside the target region
and attenuated scatterer energy within it.
These comparisons validate the effectiveness and robustness of the proposed algorithm in practical imaging applications.

Fig. \ref{Fig_exp_box_results2}
presents the imaging results
for a multi-target scenario involving four cylindrical targets.
As shown in Figs. \ref{Fig_exp_box_results2} (a) and (b),
the proposed IAA-DA algorithm successfully
resolves targets separated by only $7.6 \, \text{cm}$,
thereby demonstrating its capability for
\emph{centimeter-level imaging resolution}.
Moreover,
according to Proposition \ref{proposition_resolution},
this result yields ${{\kappa }_{I}}{{\kappa }_{S}} = 0.8 <1$,
indicating that the IAA algorithm effectively enhances
the imaging resolution.
Notably, all four targets exhibit significant energy and are reliably detected.
In contrast, the performance of the SRDI and FBP algorithms is considerably lower.
As shown in Fig. \ref{Fig_exp_box_results2} (c),
the SRDI algorithm suffers from severe energy leakage,
making it difficult to distinguish the four targets.
Similarly,
the FBP algorithm in Fig. \ref{Fig_exp_box_results2} (d)
leads to energy degradation.
Only two of the four targets maintain sufficient signal intensity,
while the remaining two are barely discernible due to their low energy levels.
These results validate the robustness and effectiveness of the proposed algorithm in multi-target imaging scenarios.

Fig. \ref{Fig_Imaging_Result_rectangle_901sdfds80270360}
demonstrates the
penetration imaging capability
of our proposed BI imaging scheme.
The imaging scenarios are given in
Figs. \ref{Fig_Imaging_Result_rectangle_901sdfds80270360} (a) and (c),
while the corresponding imaging results are presented in
Figs. \ref{Fig_Imaging_Result_rectangle_901sdfds80270360} (b) and (d),
respectively.
As observed from the results,
the ISAC system successfully images the cylindrical targets inside the box,
and the sizes of the reconstructed targets closely match their actual sizes.
These results verify that mmWave communication signals can penetrate the box and enable accurate penetration imaging, which opens up new possibilities for non-contact, high-precision internal imaging of objects in the future.

\begin{figure}[!t]
\centering
\includegraphics[width=2.83in]{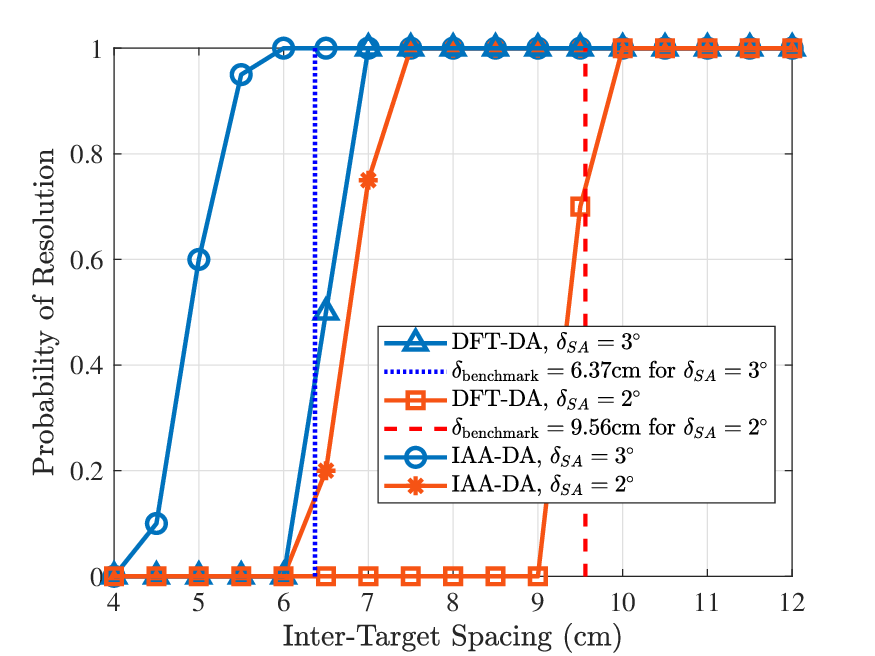}
\caption{Experimental results of the imaging resolution.}
\label{Fig_exp_resolution}
\end{figure}

Fig. \ref{Fig_exp_resolution}
presents the
experimental resolution probability
as a function of
inter-target spacing.
To calculate this probability,
two point targets
are used
and 20 independent trials are conducted for each spacing.
In particular,
the two point targets can be approximately treated as isotropic scatterers,
i.e., ${{\kappa }_{I}} \approx 1$.
Therefore,
we have
${{\delta }_{\text{DFT-DA}}}={{\kappa }_{I}}\delta _{\text{benchmark}} \approx \delta _{\text{benchmark}}$,
meaning that
the benchmark resolution can be used as a reference to assess the
resolution of the DFT-DA algorithm.
As observed, the experimental results of the DFT-DA algorithm for different $\delta_{SA}$
align well with
the theoretical predictions,
which confirms the correctness of the derivation.
In addition,
the results indicate that a larger $\delta_{SA}$
leads to improved imaging resolution.
However,
$\delta_{SA}$
should be chosen with care based on target characteristics and should remain relatively small to maintain angular coherence.
Fig. \ref{Fig_exp_resolution} also compares the resolution probabilities of the DFT-DA and IAA-DA algorithms.
Specifically,
at $\delta_{SA} = 2{}^\circ$,
the IAA-DA algorithm achieves a 1.36-fold improvement over the DFT-DA algorithm,
and a 1.38-fold improvement at $\delta_{SA} = 3{}^\circ$.
These results demonstrate the superiority of the IAA algorithm in resolution enhancement.

\section{Conclusion}\label{Section_conclusion}

In this paper,
we proposed a novel BI imaging scheme
applicable to both SAR and ISAR scenarios.
Different from conventional SAR and ISAR imaging scheme,
under which Gigahertz-level bandwidth is necessary for high-resolution range estimation in one dimension,
our proposed scheme merely requires Doppler information,
and can
work with current cellular signals with Megahertz-level bandwidth.
We also derived the imaging resolution for the proposed BI imaging scheme.
To the best of our knowledge,
our paper is the first to provide real-world experiments that rely on cellular signals for high-resolution imaging.
Meanwhile, we also conducted experiments to demonstrate the potential of utilizing millimeter-wave signals for accurate imaging of objects inside a box.

\begin{appendices}

\section{Proof of Proposition \ref{proposition_UIR}}\label{Appendix_proposition_UIR}

Note that the signal model in (\ref{sdlsdkdkjci})
is a typical multi-tone signal
at each azimuth angle $\Phi_k$.
In particular,
$-x\sin {{\Phi }_{k}}+y\cos {{\Phi }_{k}}$
is related to the position of the scatterer
and is the unknown parameter to be estimated.
Hence,
based on the Nyquist-Shannon
sampling theorem,
the maximum unambiguous frequency region
at the azimuth angle $\Phi_k$
is given by
\begin{align}
\nonumber \left| -x\sin {{\Phi }_{k}}+y\cos {{\Phi }_{k}} \right| & =\sqrt{{{x}^{2}}+{{y}^{2}}}\left|\cos \left( {{\Phi }_{k}}+{{\varphi }_{x,y}} \right) \right|\\
 &\le \frac{\lambda }{4\omega {{T}_{0}}}\sqrt{1+\frac{{{H}^{2}}}{{{R}^{2}}}}, \ \forall k,
\end{align}
where
${{\varphi }_{x,y}}=\arccos \frac{y}{\sqrt{{{x}^{2}}+{{y}^{2}}}}$.
For the above relation to hold for arbitrary scatterer locations and observation angles, the following constraint is required
\begin{align}
\sqrt{x^{2}+y^{2}}\le \frac{\lambda }{4{{\omega }}{{T}_{0}}}\sqrt{1+\frac{H^{2}}{R^{2}}}.
\end{align}
Proposition \ref{proposition_UIR}
is thus proved.

\section{Proof of the Approximation in (\ref{sdflkjslfkjsldkjf})}\label{Appendix_sdflkjslfkjsldkjf}

First, we consider the case of two scatterers.
Let
${{\tilde{f}}_{2D,n}}\left( {{x}},{{y}};\Phi  \right)={{\tilde{f}}_{n}}\left( \Phi  \right)\delta \left( x-{{x}_{n}},y-{{y}_{n}} \right)$
denote the scattering function of the $n$-th scatterer observed at the azimuth angle $\Phi_k$,
where $n=0,1$,
$\left(x_n, y_n\right)$ is the position of the $n$-th scatterer,
and $\delta \left( x, y \right)$ is the Dirac delta function.
Then, (\ref{dlijfldkjf})
can be written as
\begin{align}\label{sdfjskldkfjslkjf}
 {{g}}^{+}  \left( {D_i}, \Phi_k  \right)    =  \left| \sum\limits_{n=0}^{1}{    {{{\tilde{f}}}_{n}}\left( {{\Phi }_{k}} \right)    \chi     \left( {D_i}  -    {{D_f}}\left( {{x}_{n}},{{y}_{n}};{{\Phi }_{k}} \right) \right)} \right|,  \forall k, i.
\end{align}
Note that the sidelobes of
the ambiguity function
${{\chi }}\left( D \right)$
are generally too low to be retained
after applying the IAA algorithm.
Therefore,
outside the mainlobe region $D\in \left( -{{D}_{\text{ML}}},{{D}_{\text{ML}}} \right)$,
${{\chi }}\left( D \right)$ can be approximated as zero,
where
${{D}_{\text{ML}}}$
is the mainlobe range.
Hence,
when the interval
between the Doppler frequencies of two scatterers satisfies
$\left| {{D_f}}\left( {{x}_{1}},{{y}_{1}};{{\Phi }_{k}} \right)-{{D_f}}\left( {{x}_{0}},{{y}_{0}};{{\Phi }_{k}} \right) \right|\ge 2{{D}_{\text{ML}}}$,
(\ref{sdfjskldkfjslkjf})
can be readily approximated as
\begin{align}\label{sdfjskldkfjslksdfsdfjf}
\nonumber {{g}}^{+}  \left( {D_i}, \Phi_k  \right)  \approx \,&    \sum\limits_{n=0}^{1}{    \left| {{{\tilde{f }}}_{n}}    \left( {{\Phi }_{k}} \right) \right|{{\chi }^{+}}    \left( {D_i}   -    {{D_f}}    \left( {{x}_{n}},{{y}_{n}};{{\Phi }_{k}} \right) \right)},\\
& \forall k, i,
\end{align}
where ${{\chi }}^{+}\left( D \right)=\left|{{\chi }}\left( D \right)\right|$.
When satisfying the relationship of $\left| {{D_f}}\left( {{x}_{1}},{{y}_{1}};{{\Phi }_{k}} \right)-{{D_f}}\left( {{x}_{0}},{{y}_{0}};{{\Phi }_{k}} \right) \right| < 2{{f}_{\text{ML}}}$,
the two scatterers become unresolvable
and can generally be regarded as one scatterer.
Define ${{D}}\left( {{\Phi }_{k}} \right)\triangleq {{D_f}}\left( {{x}_{0}},{{y}_{0}};{{\Phi }_{k}} \right)\approx {{D_f}}\left( {{x}_{1}},{{y}_{1}};{{\Phi }_{k}} \right)$.
We approximate
(\ref{sdfjskldkfjslkjf})
as
\begin{align}\label{sdfjsjhlkldkfjslksdfsdfjf}
{{g}}^{+}\left( {D_i}, \Phi_k  \right)  \approx   \left| \sum\limits_{n=0}^{1}{{{{\tilde{f }}}_{n}}\left( {{\Phi }_{k}} \right)} \right|{{\chi }^{+}}\left( {D_i}-{{D_f}}\left( {{\Phi }_{k}} \right) \right), \ \forall k, i.
\end{align}
Finally,
combining the results from
(\ref{sdfjskldkfjslksdfsdfjf}) and (\ref{sdfjsjhlkldkfjslksdfsdfjf}),
we obtain
\begin{align}
{{g}}^{+}  \left( {D_i}, \Phi_k  \right)  \approx    \sum\limits_{n=0}^{1}{\tilde{f }_{n}^{+}\left( {{\Phi }_{k}} \right){{\chi }^{+}}     \left( {D_i}    -  {{D_f}}    \left( {{x}_{n}},{{y}_{n}};{{\Phi }_{k}} \right) \right)}, \forall k, i,
\end{align}
where
$\tilde{f }_{n}^{+}\left( {{\Phi }_{k}} \right)\in \mathbb{R}$
is the equivalent scattering function.
The above results
can be readily extended to multi-scatterer cases
using the same method above.
Therefore, the approximation in (\ref{sdflkjslfkjsldkjf})
is proved.

\section{Proof of Proposition \ref{proposition_resolution}}\label{Appendix_proposition_resolution}
We begin by deriving the resolution benchmark,
i.e., $\delta_{\text{benchmark}}$.
This benchmark is established
under the assumption of an isotropic equivalent scattering function
and is computed by employing the DFT-DA imaging algorithm.
Then, we prove that
the anisotropic equivalent scattering function considered in this paper
leads to worse resolution performance,
i.e., $\kappa_I \ge 1$.
Additionally, for normalization, we define $\text{PSF}(0,0) = 1$ in accordance with Definition \ref{Definition_resolution}.

For an
isotropic equivalent scattering function,
denoted by $\tilde{A }_{2D}^{+}$,
according to (\ref{Dfsdlfksl}) and the normalization condition,
we have
$\text{PSF}_\text{iso}\left( 0,0 \right)=\frac{1}{2\pi }\int_{0}^{2\pi }{\tilde{A }_{2D}^{+}d\Phi } = 1$,
which means $\tilde{A }_{2D}^{+} = 1$.
Therefore, the PSF in (\ref{jsflkdjlfskjlkd})
can be expressed as
\begin{align}\label{dsflksdjhflksj}
  & \text{PSF}_\text{iso}\left( x,y;\beta  \right)  =  \frac{1}{2\pi }  \int_{0}^{2\pi }  {{{\chi }^{+}}  \left( {{\beta }^{-1}}  \sqrt{{{x}^{2}}  +  {{y}^{2}}}   \cos \Phi  \right)   d\Phi },
\end{align}
where $\beta$ is defined in (\ref{jsflkdjlfskjlkd}).
Note that
\begin{align}\label{sdlfksjlfkj}
\text{PSF}_\text{iso}\left( a\sin  \alpha ,a\cos  \alpha ;  \beta  \right)  =  \frac{1}{2\pi }  \int_{0}^{2\pi }{{{\chi }^{+}}\left( {{\beta }^{-1}}a\cos  \Phi  \right)  d\Phi },
\end{align}
where
$a\in {{\mathbb{R}}_{\ge 0}}$,
and
$\alpha\in \left[0, 2 \pi\right]$ denotes the inclination angle of the slice.
The right-hand side of (\ref{sdlfksjlfkj})
is unrelated to $\alpha$.
Therefore,
slices of the PSF passing through the origin
exhibit identical spectrums at arbitrary inclination angles,
which means
that
the $-3\text{dB}$ contour of the normalized PSF,
$\zeta\left(\alpha\right)$,
as defined in Definition \ref{Definition_resolution},
remains constant with respect to the angle $\alpha$,
i.e., ${{\zeta}_{\text{iso}}}\triangleq \zeta\left( \alpha  \right)$.

Next, we use the resolution of the DFT-DA algorithm
as a reference
benchmark as it provides the worst achievable resolution.
Note that $\zeta\left(\alpha\right) + \zeta\left(\alpha + \pi\right) = 2{{\zeta}_{\text{iso}}}$
is the imaging resolution
in the case of the isotropic equivalent scattering function.
We need to determine ${{\zeta}_{\text{iso}}}$
such that $\text{PSF}_\text{iso}\left( {{\zeta}_{\text{iso}}}\sin \alpha ,{{\zeta}_{\text{iso}}} \cos \alpha ;\beta  \right) = -3\text{dB}$.
Here, we first
derive the resolution for the case where $\beta = 1$,
and then extend the result to general form in (\ref{sdlfksjlfkj}).
Note that
the resolution of
${{\chi }_{\text{DFT}}^{+}}\left( {\cdot} \right)$
is $\frac{1}{\left( 2M+1 \right){{T}_{0}}}$
according to
the principle of the
digital signal processing.
However, it is quite difficult to obtain a precise resolution expression.
Based on Definition \ref{Definition_resolution},
numerical analysis reveals that when $\beta = 1$,
the resolution of the PSF can approximately be expressed as
\begin{align}\label{sdlfjksjlfkj}
{{\zeta}_{\text{iso}}^{\beta = 1}} \approx \frac{{C_R}/{2}}{\left( 2M+1 \right){{T}_{0}}},
\end{align}
where $C_R \approx 1.29$ is a constant.
Due to
$\text{PSF}_\text{iso}\left( a\sin  \alpha , a \cos  \alpha;1 \right)    =    \text{PSF}_\text{iso}\left( \beta a \sin  \alpha , \beta a \cos  \alpha;\beta  \right)$,
the benchmark resolution
is given by
\begin{align}
\nonumber \delta_{\text{benchmark}}&\triangleq  2{{\zeta}_{\text{iso}}}= 2 \beta {{\zeta}_{\text{iso}}^{\beta = 1}}
\approx \beta \frac{C_R}{\left( 2M+1 \right){{T}_{0}}}  \\
&\approx C_R \frac{\lambda }{4\omega {{T}_{0}}M}\sqrt{1+\frac{{{H}^{2}}}{{{R}^{2}}}}  = C_R\frac{\lambda }{4{{\delta }_{SA}}}\sqrt{1+\frac{{{H}^{2}}}{{{R}^{2}}}}.
\end{align}
Therefore,
the benchmark resolution $\delta_{\text{benchmark}}$
in (\ref{dsyysdsdsfkjsddfsdfslkfj})
is proved.
Because the IAA algorithm provides
super-resolution capability,
the
IAA-DA algorithm yields lower resolution
compared to the DFT-DA algorithm,
meaning that
$\kappa_S \le 1$.

Next, we evaluate the resolution under the anisotropic equivalent scattering function
and focus on $\kappa_I$.
For anisotropic equivalent scattering function,
the PSF in (\ref{jsflkdjlfskjlkd})
can be expressed as
\begin{align}
\nonumber &\text{PSF}\left( a\sin \alpha ,a\cos \alpha ;\beta  \right)\\
& \ \ \ \ \ \ \ \ \    =\frac{1}{2\pi }\int_{0}^{2\pi }{\tilde{A }_{2D}^{+}\left( \Phi -\alpha  \right){{\chi }^{+}}\left( {{\beta }^{-1}}a\cos \Phi  \right)d\Phi }.
\end{align}
We now
define
the function
$h\left(\alpha\right)$
to quantify the difference between the anisotropic and isotropic PSFs along the benchmark contour:
$h\left(\alpha\right) \triangleq  \text{PSF}\left( {{\zeta}_{\text{iso}}}\sin \alpha ,{{\zeta}_{\text{iso}}}\cos \alpha ;\beta  \right) - \text{PSF}_\text{iso}\left( {{\zeta}_{\text{iso}}}\sin \alpha ,{{\zeta}_{\text{iso}}}\cos \alpha ;\beta  \right)$.
Note that the following relationship holds
\begin{align}\label{ewlekrjwkdn}
\nonumber\frac{1}{2\pi }\int_{0}^{2\pi }{h\left( \alpha  \right)d\alpha } = {}& \frac{1}{2\pi }\int_{0}^{2\pi }{\left[ \frac{1}{2\pi }\int_{0}^{2\pi }{\tilde{A }_{2D}^{+}\left( \Phi -\alpha  \right)d\alpha }-1 \right] } \\
& \times {{\chi }^{+}}\left( \frac{{C_R}\cos \Phi }{2\left( 2M+1 \right){{T}_{0}}} \right)d\Phi =0.
\end{align}
This
implies that there exists $\alpha_0 \in [0, 2\pi]$
such that $h\left({\alpha_0}\right) > 0$,
meaning that
$\text{PSF}\left( {{\zeta}_{\text{iso}}}\sin \alpha_0 ,{{\zeta}_{\text{iso}}}\cos \alpha_0 ;\beta  \right) > -3\text{dB}$.
Then,
note that
for any $\alpha \in [0, 2\pi]$,
$\text{PSF}(a\sin \alpha, a\cos \alpha; \beta)$
is a monotonically decreasing function of $a$
over $a \in \left[0, {{\zeta}_{\text{iso}}}\right]$.
Therefore,
the value of $a$ should exceed ${{\zeta}_{\text{iso}}}$
in order for the function to fall to $-3\text{dB}$ required in Definition \ref{Definition_resolution}.
Hence, the anisotropy of the equivalent scattering function
leads to a poor resolution, which means $\kappa_I \ge 1$.
Proposition \ref{proposition_resolution} is thus proved.

\end{appendices}

\bibliographystyle{IEEEtran}
\bibliography{AAA}

@ARTICLE{650078,
  author={Trintinalia, L.C. and Bhalla, R. and Hao Ling},
  journal={IEEE Trans. Antennas Propag.},
  title={Scattering center parameterization of wide-angle backscattered data using adaptive {Gaussian} representation},
  year={Nov. 1997},
  volume={45},
  number={11},
  pages={1664-1668},
  doi={10.1109/8.650078}}

@ARTICLE{805442,
  author={Li Xi and Liu Guosui and Jinlin Ni},
  journal={IEEE Trans. Aerosp. Electron. Syst.},
  title={Autofocusing of {ISAR} images based on entropy minimization},
  year={Oct. 1999},
  volume={35},
  number={4},
  pages={1240-1252},
  doi={10.1109/7.805442}}

@ARTICLE{532283,
  author={Itoh, T. and Sueda, H. and Watanabe, Y.},
  journal={IEEE Trans. Aerosp. Electron. Syst.},
  title={Motion compensation for {ISAR} via centroid tracking},
  year={Jul. 1996},
  volume={32},
  number={3},
  pages={1191-1197},
  doi={10.1109/7.532283}}

@ARTICLE{4497843,
  author={Dehmollaian, Mojtaba and Sarabandi, Kamal},
  journal={IEEE Trans. Geosci. Remote Sens.},
  title={Refocusing Through Building Walls Using Synthetic Aperture Radar},
  year={Jun. 2008},
  volume={46},
  number={6},
  pages={1589-1599},
  doi={10.1109/TGRS.2008.916212}}

@ARTICLE{755021,
  author={Carin, L. and Geng, N. and McClure, M. and Sichina, J. and Lam Nguyen},
  journal={IEEE Antennas Propag. Mag.},
  title={Ultra-wide-band synthetic-aperture radar for mine-field detection},
  year={Feb. 1999},
  volume={41},
  number={1},
  pages={18-33},
  doi={10.1109/74.755021}}

@ARTICLE{9534682,
  author={Tong, Xin and Zhang, Zhaoyang and Wang, Jue and Huang, Chongwen and Debbah, M¨¦rouane},
  journal={IEEE J. Sel. Topics Signal Process.},
  title={Joint Multi-User Communication and Sensing Exploiting Both Signal and Environment Sparsity},
  year={Nov. 2021},
  volume={15},
  number={6},
  pages={1409-1422},
  doi={10.1109/JSTSP.2021.3111432}}

@ARTICLE{10412121,
  author={Baczyk, Marcin Kamil and Samczynski, Piotr and Drozdowicz, Jedrzej and Wielgo, Maciej and Sobolewski, Jakub and Ciesielski, Marek and Julczyk, Jakub and Stasiak, Krzysztof and Pietrzykowski, Grzegorz and Abratkiewicz, Karol and Soszka, Maciej},
  journal={IEEE J. Sel. Topics Appl. Earth Observ. Remote Sens.},
  title={{3-D} High-Resolution {ISAR} Imaging for Noncooperative Air Targets},
  year={2024},
  volume={17},
  number={},
  pages={4194-4207},
  doi={10.1109/JSTARS.2024.3357120}}

@ARTICLE{942570,
  author={Sheen, D.M. and McMakin, D.L. and Hall, T.E.},
  journal={IEEE Trans. Microwave Theory Tech.},
  title={Three-dimensional millimeter-wave imaging for concealed weapon detection},
  year={Sep. 2001},
  volume={49},
  number={9},
  pages={1581-1592},
  doi={10.1109/22.942570}}

@ARTICLE{5420035,
  author={Potter, Lee C. and Ertin, Emre and Parker, Jason T. and Cetin, M¨¹jdat},
  journal={Proc. IEEE},
  title={Sparsity and Compressed Sensing in Radar Imaging},
  year={Jun. 2010},
  volume={98},
  number={6},
  pages={1006-1020},
  doi={10.1109/JPROC.2009.2037526}}

@ARTICLE{10637442,

  author={Hu, Yanmo and Andrew Zhang, J. and Wu, Kai and Deng, Weibo and Jay Guo, Y.},

  journal={IEEE Trans. Commun.},

  title={Anchor Points Assisted Uplink Sensing in Perceptive Mobile Networks},

  year={Feb. 2025},

  volume={73},

  number={2},

  pages={904-920},

  doi={10.1109/TCOMM.2024.3443729}}

@book{CT_boooook,
  author = {J. Hsieh},
  title = {Computed Tomography: Principles, Design, Artifacts, and Recent Advances},
  address = {Bellingham, WA, USA: SPIE, 2009},
}

@ARTICLE{3135465484121,

  author={R. N. Bracewell and A. C. Riddle},

  journal={Astrophys. J.},

  title={Inversion of fan-beam scans in radio astronomy},

  year={Nov. 1967},

  volume={150},

  number={},

  pages={427-434}}

@ARTICLE{5417172,

  author={Yardibi, Tarik and Li, Jian and Stoica, Petre and Xue, Ming and Baggeroer, Arthur B.},

  journal={IEEE Trans. Aerosp. Electron. Syst.},

  title={Source Localization and Sensing: A Nonparametric Iterative Adaptive Approach Based on Weighted Least Squares},

  year={Jan. 2010},

  volume={46},

  number={1},

  pages={425-443},

  doi={10.1109/TAES.2010.5417172}}

@ARTICLE{10851319,

  author={Zhang, Ruiyun and Wang, Zhaolin and Wei, Zhiqing and Liu, Yuanwei and Xiong, Zehui and Feng, Zhiyong},

  journal={ IEEE Trans. Cogn. Commun. Netw.},

  title={{ISAR} Sensing Based on {MUSIC} Algorithm in Integrated Sensing and Communications},

  year={Oct. 2025},

  volume={11},

  number={5},

  pages={3292-3305},

  doi={10.1109/TCCN.2025.3533000}}

@ARTICLE{11288092,
  author={Gao, Junyuan and Zhu, Weifeng and Zhang, Shuowen and Wu, Yongpeng and Cao, Jiannong and Caire, Giuseppe and Liu, Liang},
  journal={IEEE Trans. Wireless Commun.},
  title={Integrated Massive Communication and Target Localization in {6G} Cell-Free Networks},
  year={2026},
  volume={25},
  number={},
  pages={8498-8515},
  doi={10.1109/TWC.2025.3638614}}

@ARTICLE{9724258,

  author={Shi, Qin and Liu, Liang and Zhang, Shuowen and Cui, Shuguang},

  journal={IEEE J. Sel. Areas Commun.},

  title={Device-Free Sensing in {OFDM} Cellular Network},

  year={Jun. 2022},

  volume={40},

  number={6},

  pages={1838-1853},

  doi={10.1109/JSAC.2022.3155543}}

@ARTICLE{10878492,
  author={Gonzalez-Prelcic, Nuria and Tagliaferri, Dario and Keskin, Musa Furkan and Wymeersch, Henk and Song, Lingyang},
  journal={IEEE Veh. Technol. Mag.},
  title={Six Integration Avenues for {ISAC} in {6G} and Beyond},
  year={Mar. 2025},
  volume={20},
  number={1},
  pages={18-39},
  doi={10.1109/MVT.2025.3529403}}

@ARTICLE{10367810,
  author={Dubey, Amartansh and Li, Zan and Murch, Ross},
  journal={IEEE Trans. Antennas Propag.},
  title={Reconciling Radio Tomographic Imaging With Phaseless Inverse Scattering},
  year={Feb. 2024},
  volume={72},
  number={2},
  pages={1837-1849},
  doi={10.1109/TAP.2023.3342827}}

@ARTICLE{10618967,
  author={Li, Zan and Dubey, Amartansh and Shen, Shanpu and Kundu, Neel Kanth and Rao, Junhui and Murch, Ross},
  journal={IEEE Trans. Wireless Commun.},
  title={Radio Tomographic Imaging With Reconfigurable Intelligent Surfaces},
  year={Nov. 2024},
  volume={23},
  number={11},
  pages={15784-15797},
  doi={10.1109/TWC.2024.3433011}}

@ARTICLE{10947014,
  author={Li, Zan and Dubey, Amartansh and Murch, Ross},
  journal={IEEE Open J. Antennas Propag.},
  title={The Impact of Phase Information on the Imaging Performance of the Extended Phaseless {Rytov} Approximation},
  year={Aug. 2025},
  volume={6},
  number={4},
  pages={989-1000},
  doi={10.1109/OJAP.2025.3556848}}

@ARTICLE{5374407,
  author={Wilson, Joey and Patwari, Neal},
  journal={IEEE Trans. Mobile Comput.},
  title={Radio Tomographic Imaging with Wireless Networks},
  year={May 2010},
  volume={9},
  number={5},
  pages={621-632},
  doi={10.1109/TMC.2009.174}}

@ARTICLE{752218,
  author={Sato, T.},
  journal={IEEE Trans. Geosci. Remote Sens.},
  title={Shape estimation of space debris using single-range {Doppler} interferometry},
  year={Mar. 1999},
  volume={37},
  number={2},
  pages={1000-1005},
  doi={10.1109/36.752218}}

@ARTICLE{10475383,
  author={Huang, Yixuan and Yang, Jie and Wen, Chao-Kai and Jin, Shi},
  journal={IEEE Trans. Commun.},
  title={{RIS}-Aided Single-Frequency {3D} Imaging by Exploiting Multi-View Image Correlations},
  year={Aug. 2024},
  volume={72},
  number={8},
  pages={5003-5018},
  doi={10.1109/TCOMM.2024.3379351}}

@ARTICLE{4502060,
  author={Ausherman, Dale A. and Kozma, Adam and Walker, Jack L. and Jones, Harrison M. and Poggio, Enrico C.},
  journal={IEEE Trans. Aerosp. Electron. Syst.},
  title={Developments in Radar Imaging},
  year={Jul. 1984},
  volume={AES-20},
  number={4},
  pages={363-400},
  doi={10.1109/TAES.1984.4502060}}

@ARTICLE{10663814,
  author={Liu, Liang and Zhang, Shuowen and Cui, Shuguang},
  journal={IEEE Commun. Mag.},
  title={Leveraging a Variety of Anchors in Cellular Network for Ubiquitous Sensing},
  year={Sep. 2024},
  volume={62},
  number={9},
  pages={98-104},
  doi={10.1109/MCOM.001.2300632}}

@ARTICLE{10948152,
  author={Zhu, Weifeng and Zhang, Shuowen and Liu, Liang},
  journal={IEEE Trans. Wireless Commun.},
  title={Joint Transmission and Compression Optimization for Networked Sensing With Limited-Capacity Fronthaul Links},
  year={Aug. 2025},
  volume={24},
  number={8},
  pages={6643-6657},
  doi={10.1109/TWC.2025.3555005}}

@ARTICLE{11079818,
  author={Hu, Yanmo and Wu, Kai and Zhang, J. Andrew and Deng, Weibo and Jay Guo, Y.},
  journal={IEEE Trans. Wireless Commun.},
  title={Cross-Frequency Sensing in Bistatic {ISAC} Systems},
  year={Jan. 2026},
  volume={25},
  number={},
  pages={681-697},
  doi={10.1109/TWC.2025.3585944}}

@ARTICLE{4330we23965,
  author={C. C. Aleksoff and C. R. Christensen},
  journal={Appl. Opt.},
  title={Holographic {Doppler} imaging of rotating objects},
  year={Jan. 1975},
  volume={14},
  number={1},
  pages={134-141},
  doi={}}

@ARTICLE{6153063,
  author={Zhuge, Xiaodong and Yarovoy, Alexander G.},
  journal={IEEE Trans. Image Process.},
  title={Three-Dimensional Near-Field {MIMO} Array Imaging Using Range Migration Techniques},
  year={Feb. 2012},
  volume={21},
  number={6},
  pages={3026-3033},
  doi={10.1109/TIP.2012.2188036}}

@misc{gao2025covariancebasedimagingmultiviewfusion,
      title={Covariance-based Imaging and Multi-View Fusion for Networked Sensing},
      author={Junyuan Gao and Weifeng Zhu and Yanmo Hu and Shuowen Zhang and Jiannong Cao and Yongpeng Wu and Giuseppe Caire and Liang Liu},
      year={2025},
      eprint={2511.14490},
      archivePrefix={arXiv},
      primaryClass={eess.SP},
      url={https://arxiv.org/abs/2511.14490},
}

@ARTICLE{6504845,
  author={Moreira, Alberto and Prats-Iraola, Pau and Younis, Marwan and Krieger, Gerhard and Hajnsek, Irena and Papathanassiou, Konstantinos P.},
  journal={IEEE Geosci. Remote Sens. Mag.},
  title={A tutorial on synthetic aperture radar},
  year={Mar. 2013},
  volume={1},
  number={1},
  pages={6-43},
  doi={10.1109/MGRS.2013.2248301}}

@ARTICLE{4373378,
  author={Meta, Adriano and Hoogeboom, Peter and Ligthart, Leo P.},
  journal={IEEE Trans. Geosci. Remote Sens.},
  title={Signal Processing for {FMCW} {SAR}},
  year={Nov. 2007},
  volume={45},
  number={11},
  pages={3519-3532},
  doi={10.1109/TGRS.2007.906140}}

@BOOK{ArraySignalProcessing,

  author={Harry L. Van Trees},

  title={Optimum Array Processing. Detection, Estimation, and Modulation Theory, Part IV},
  
  year = {},
  
  publisher={},

  address = {John Wiley \& Sons, Inc., New York, USA, 2002},

}

@ARTICLE{317861,
  author={Mathews, C.P. and Zoltowski, M.D.},
  journal={IEEE Trans. Signal Process.},
  title={Eigenstructure techniques for {2-D} angle estimation with uniform circular arrays},
  year={Sep. 1994},
  volume={42},
  number={9},
  pages={2395-2407},
  doi={10.1109/78.317861}}

@INPROCEEDINGS{yanmo_conference,
  author={Hu, Yanmo and Zhang, Shuowen and Murch, Ross and Liu, Liang},
  booktitle={Proc. IEEE Int. Conf. Commun. (ICC)},
  title={Bandwidth-Independent Imaging in {6G} Integrated Sensing and Communication Systems},
  year={2026},
  volume={},
  number={},
  pages={}}
\vfill

\end{document}